\documentclass[11pt,english]{article}
\pdfoutput=1
\usepackage[unicode=true,
bookmarks=true,
bookmarksnumbered=false,
bookmarksopen=false,
breaklinks=false,
pdfborder={0 0 1},
backref=slide,
colorlinks=true,
linktoc=all,
citecolor=blue]
{hyperref}

\usepackage{mathpazo}
\usepackage[T1]{fontenc}
\usepackage[latin9]{inputenc}
\usepackage[active]{srcltx}
\usepackage{babel}
\usepackage{verbatim}
\usepackage{amsmath}
\usepackage{amsthm}
\usepackage{amssymb}
\usepackage{wasysym}

\usepackage{color}

\makeatletter
\theoremstyle{plain}
\newtheorem{thm}{\protect\theoremname}
\theoremstyle{plain}
\newtheorem{cor}[thm]{\protect\corollaryname}
\theoremstyle{definition}
\newtheorem{defn}[thm]{\protect\definitionname}
\theoremstyle{remark}
\newtheorem{rem}[thm]{\protect\remarkname}
\theoremstyle{plain}
\newtheorem{lem}[thm]{\protect\lemmaname}

\usepackage{fullpage}

\usepackage{graphicx}

\makeatother

\providecommand{\corollaryname}{Corollary}
\providecommand{\definitionname}{Definition}
\providecommand{\lemmaname}{Lemma}
\providecommand{\remarkname}{Remark}
\providecommand{\theoremname}{Theorem}

\begin{document}

\title{Tree tribes and lower bounds for switching lemmas}

\author{Jenish C. Mehta \thanks{California Institute of Technology.}}

\maketitle
	
\begin{abstract}
We show tight upper and lower bounds for switching lemmas obtained
by the action of random $p$-restrictions on boolean functions that
can be expressed as decision trees in which every vertex is at a distance
of at most $t$ from some leaf, also called $t$-clipped decision
trees. More specifically, we show the following:
\begin{enumerate}
\item If a boolean function $f$ can be expressed as a $t$-clipped decision
tree, then under the action of a random $p$-restriction $\rho$,
the probability that the smallest depth decision tree for $f|_{\rho}$
has depth greater than $d$ is upper bounded by $(4p2^{t})^{d}$. {\small \par}
\item For every $t$, there exists a function $g_{t}$ that can be expressed
as a $t$-clipped decision tree, such that under the action of a random
$p$-restriction $\rho$, the probability that the smallest depth
decision tree for $g_{t}|_{\rho}$ has depth greater than $d$ is
\emph{lower} bounded by $(c_{0}p2^{t})^{d}$, for $0\leq p\leq c_{p}2^{-t}$
and $0\leq d\leq c_{d}\frac{\log n}{2^{t}\log t}$, where $c_{0},c_{p},c_{d}$
are universal constants. {\small \par}
\end{enumerate}
\end{abstract}

\newpage

\tableofcontents{}

\newpage

\section{Introduction}

One useful and powerful idea to separate a boolean function $g$ from
some set ${\cal F}$ of boolean functions over $\{0,1\}^{n}$ is to
use \emph{restrictions}. By showing that restricted to some subset
of $\{0,1\}^{n}$, the functions in ${\cal F}$ become \emph{simple},
but the function $g$ does not become simple, it can be concluded
that $g\not\in{\cal F}$. The extent to which functions in $\mathcal{F}$
become simple is captured by switching lemmas, which informally try
to answer the following question: Given a family ${\cal F}$ of boolean
functions characterized by some parameter $t$, and a family ${\cal S}$
of distributions over subsets of $\{0,1\}^{n}$ characterized by some
parameter $p$, how complex does a function $f\in{\cal F}$ remain
after it is restricted to a subset $S$ chosen according to some distribution
$D\in{\cal S}$? Defining the measure of complexity of a function
and the sets ${\cal F}$ and ${\cal S}$ gives a switching lemma of
a particular type. 

Switching lemmas originated with the works of \cite{furst1984parity}
and \cite{ajtai198311}, and were proved in their strongest form for
DNFs by Hastad \cite{hastad1986almost}, which answered the question
stated above as follows: Let ${\cal F}$ be the set of DNFs (or CNFs)
of width $t$ over $n$ boolean variables. Let ${\cal S}$ contain
exactly one distribution over subsets of $\{0,1\}^{n}$, also called
a random $p$-restriction, which chooses the subset $S$ as $S=S_{1}\oplus\ldots\oplus S_{n}$
where $\oplus$ is the direct sum, by independently choosing each
$S_{i}$ as the set $\{0,1\}$ with probability $p$, the set $\{0\}$
with probability $\frac{1-p}{2}$, and the set $\{1\}$ with probability
$\frac{1-p}{2}$. The measure of complexity of a function is the depth
of the smallest depth decision tree that decides it. With these instantiations,
it was shown in \cite{hastad1986almost} that the probability that
the decision tree for a $t$-DNF has depth greater than $d$ after
it is restricted to a random $p$-restriction is upper bounded by
$(5pt)^{d}$.

The original proof of this result in \cite{hastad1986almost} used
conditioning on values of the variables under the applied restriction,
and later \cite{razborov1993equivalence} gave an alternate combinatorial
proof of the same fact. An excellent survey and explanations of many
of these results is in \cite{beame1994switching}. Apart from giving
an alternate proof that parity is not in $AC_{0}$, switching lemmas
and their variants for different families of functions and restrictions
have found a large number of applications, to obtain lower bounds
on circuit size and depth and oracle separations of complexity classes
\cite{sipser1983borel,ajtai198311,furst1984parity,yao1985separating,hastad1986almost,lynch1986depth,cai1989probability,ajtai1989first,beame1990lower,rossman2015average,haastad2016average},
and limitations on bounded depth proof systems \cite{ajtai1990parity,beame1992exponential,beame1993exponential,bellantoni1992approximation,krajivcek1995exponential,pitassi1993exponential,impagliazzo2001counting,pitassi2016poly},
amongst others.

\subsubsection*{Our results}

A natural generalization of $t$-DNFs are functions that can be expressed
as $t$-clipped decision trees, first introduced in \cite{pitassi2016poly}.
A $t$-clipped decision tree is a decision tree in which every vertex
is at a distance of at most $t$ from some leaf. As observed in \cite{pitassi2016poly},
every $t$-DNF can be expressed as a $t$-clipped decision tree.\footnote{Note however that the other way round is not true, specifically because
for boolean variables $a,b,x$, $a\cdot x\vee\overline{x}\cdot b\not\equiv a\cdot x\vee b$
for $a=0,x=1,b=1$. } Casting $t$-DNFs as $t$-clipped decision trees, \cite{pitassi2016poly}
proved a strong switching lemma for $t$-clipped decision trees restricted
to randomly chosen affine subspaces of $\{0,1\}^{n}$. Describing
the results in \cite{pitassi2016poly} requires considerable setup,
and we will do that in the next section. 

Our first result is an improved upper bound on the action of random
$p$-restrictions on $t$-clipped decision trees. We write $\text{DT}_{\text{depth}}(g)$
for the depth of the least depth decision tree for the function $g$,
and for simplicity, we write $f|_{\rho}$ to mean that $f$ is restricted
to a subset of $\{0,1\}^{n}$ chosen according to $\rho$. 
\begin{thm}
\label{thm:upper_bound}For any boolean function $f$ that has a $t$-clipped
decision tree, for a random $p$-restriction $\rho$, 
\[
\Pr_{\rho}[\text{DT}_{\text{depth}}(f|_{\rho})\geq d]\leq(4p2^{t})^{d}.
\]

\end{thm}
For the case of random $p$-restrictions, Theorem \ref{thm:upper_bound}
is an improvement by a factor of $(10t)^{d}$ over the bounds in \cite{pitassi2016poly}.
Our proof of Theorem \ref{thm:upper_bound} is recursive, and uses
conditioning on variables similar to that in \cite{hastad1986almost}.
We would like to remark that although switching lemmas in general
are proved in two essentially equivalent ways - one that uses conditioning
on variables and another that uses combinatorial arguments - the method
of conditioning is seen to give slightly stronger bounds, as seen
by the slightly better constant in \cite{hastad1986almost} over \cite{razborov1993equivalence},
and an improvement similar to \ref{thm:upper_bound} in \cite{haastad2016average}
over \cite{rossman2015average}. Theorem \ref{thm:upper_bound} is
another case in point for random $p$-restrictions over the result
in \cite{pitassi2016poly}.

A natural question is whether the bound in Theorem \ref{thm:upper_bound}
is tight, or can further be improved to a bound of $O(pt)^{d}$ similar
to that for $t$-DNFs. However, we show that it is asymptotically
tight up to constant factors. 
\begin{thm}
\label{thm:lower_bound}For every $t\geq1$, there is a function $g_{t}$
expressible as a $t$-clipped decision tree, such that for a random
$p$-restriction $\rho$, for $0\leq p\leq c_{p}2^{-t}$ and $0\leq d\leq c_{d}\left(\frac{\log n}{2^{t}\log t}\right)$,
\[
\Pr_{\rho}[\text{DT}_{\text{depth}}(g_{t}|_{\rho})\geq d]\geq(c_{0}p2^{t})^{d},
\]
where $c_{p}$, $c_{d}$ and $c_{0}$ are universal constants. 
\end{thm}
The functions $g_{t}$ that we create for Theorem \ref{thm:lower_bound}
are explicit constructions, and we call them tree tribes, or more
specifically, clipped xor tree tribes. It follows from Theorem \ref{thm:lower_bound}
that these functions, expressible at $t$-clipped decision trees,
are more resilient to random $p$-restrictions than $t$-DNFs. Further,
the fourier coefficients of these functions have combinatorial properties
that might be of independent interest, and we think that these functions
might serve as counterexamples or ``extreme points'' for other problems. 

The proof of Theorem \ref{thm:lower_bound} is recursive, and proceeds
by analyzing the coefficients of polynomials that arise in the analysis
of tree tribes. Note that in Theorem \ref{thm:lower_bound}, since
we want to lower bound the probability of the event $\text{DT}_{\text{depth}}(g_{t}|_{\rho})\geq d$,
we require that the decision tree with the least depth (and thus \emph{every}
decision tree) for $g_{t}|_{\rho}$ must have depth greater than $d$
with sufficient probability. To achieve this, our proof proceeds as
follows: If $T$ is the decision tree for $g_{t}$, we lower bound
the probability of finding \emph{``paths with a split''} in $T|_{\rho}$.
A \emph{``path with a split''} is a subtree of $T|_{\rho}$, which
consists of a path of distinct variables $y_{1},\ldots,y_{d}$ (where
$y_{1}$ is closest to the root in $T|_{\rho}$), such that $y_{d}$
is connected to two leaves with \emph{different }values (more specifically,
$y_{d}$ has a path to a leaf labelled 0 and a leaf labelled 1). Any
decision tree for such a subtree of $T|_{\rho}$ must have depth at
least $d$ (at least if the variables $y_{1},\ldots,y_{d}$ do not
appear elsewhere in the tree), by exactly the same argument that any
decision tree for the OR function on $d$ variables must have depth
at least $d$. However, we need a sufficient number of vertices in
the tree before we get sufficient probability mass for the event of
finding such a path with a split, and this is the reason we get an
upper bound on the depth $d$ for which Theorem \ref{thm:lower_bound}
holds. 

The application of Theorem \ref{thm:lower_bound} is to bounded depth
Frege proof systems, which we discuss next.

\subsubsection*{Implications for proof complexity }

Proof complexity, originating with \cite{cook1979relative}, tries
to answer whether $\mathsf{NP}$ is same as $\mathsf{coNP}$, starting
with the observation that if an unsatisfiable 3-SAT formula has a
short proof of unsatisfiability in \emph{any} propositional proof
system, then $\mathsf{NP}=\mathsf{coNP}$. However, since we expect
$\mathsf{NP}$ to be different from $\mathsf{coNP}$, one way to attack
the problem is to show that in many specific proof systems, some unsatisfiable
formula requires long proofs. 

A natural proof system to consider is the Frege proof system, in which
each line of the proof is an $\mathsf{AC}_{0}$ formula over boolean
variables, with the connectives $\neg,\vee,\wedge,\implies$, has
modus ponens as the only rule of inference, and has a small set of
simple axioms. In a series of works, it was shown that any constant
depth Frege proof for the Pigeonhole Principle (PHP) requires superpolynomial
size \cite{ajtai1988complexity}, any polynomial size Frege proof
for PHP requires $\Omega(\log\log n)$ depth \cite{bellantoni1992approximation,krajivcek1995exponential,pitassi1993exponential},
and similar bounds hold for Tseitin formulas (contradictions) over
the complete graph \cite{urquhart1996simplified}. Alternately, an
upper bound was shown in \cite{buss1987polynomial}, by which both
PHP and Tseitin formulas over any graph have polynomial size, $O(\log n)$
depth Frege proofs. 

This gap between the power of $O(\log\log n)$ and $O(\log n)$ depth
Frege proofs was long open, until the recent work \cite{pitassi2016poly},
in which it was shown that there is a Tseitin formula over a 3-regular
expander graph, such that any $O(\sqrt{\log n})$ depth Frege proof
for it must have super polynomial size. A crucial idea in \cite{pitassi2016poly}
is the use of random projections, which has recently been used to
obtain powerful correlation bounds within constant depth circuits
in the breakthrough work \cite{rossman2015average}, and further improved
in \cite{haastad2016average}. 

At the heart of the result in \cite{pitassi2016poly} is a delicately
constructed switching lemma, such that when applied to Tseitin formulas
over 3-regular expanders, the resulting graph remains a 3-regular
expander. An exemplary exposition of the switching lemma used in \cite{pitassi2016poly}
is given in \cite{rossman2016tutorial}, and we present the description
given there.

Let the universe $U=\{0,1\}^{n}$ where we consider the natural equivalence
between elements of $U$ and subsets of $[n]$. For $A\in U$ and
$B\subseteq U$, we say that the set $B$ shatters $A$ if for every
$A'\subseteq A$, there exists $B'\in B$ such that $B'\cap A=A'$.
Let $B$ be some affine subspace of $U$. A decision tree $T$ is
$B$-independent, if $B$ shatters the set of variables on every root
to leaf path in $T$. Let ${\cal F}_{t,B}$ be the set of all functions
for which there is a $B$-independent $t$-clipped decision tree.
We say that an arbitrary distribution $V$ over the random linear
subspaces of $B$ is $p$-bounded, if for every $J\in U$, the probability
that a subspace chosen according to $V$ shatters $J$ is upper bounded
by $p^{|J|}$. Let ${\cal S}_{p,B}$ be the set of all distributions
$W=V+u$, where $V$ is a $p$-bounded distribution and $u$ is chosen
uniformly from $B$. We write $\rho\leftarrow W$ to denote an affine
subspace chosen according to $W$. Given this setup, the following
theorem is shown in \cite{pitassi2016poly,rossman2016tutorial}. 
\begin{thm}
\label{thm:prst}\cite{pitassi2016poly} For any arbitrary affine
subspace $B$, for every $f\in{\cal F}_{t,B}$ and $W\in{\cal S}_{p,B}$,
\[
\Pr_{\rho\leftarrow W}[\text{DT}_{\text{depth}}(f|_{\rho})\geq d]\leq(40pt2^{t})^{d}.
\]

\end{thm}
The proof of Theorem \ref{thm:prst} is via a beautiful combinatorial
argument, that applies restrictions to \emph{all }the variables in
a $t$-clipped decision tree directly, and notably, does not use recursion.
The reader is referred to \cite{pitassi2016poly,rossman2016tutorial}
for the proof. Casting $t$-DNFs as $t$-clipped decision trees and
using Theorem \ref{thm:prst}, it was shown in \cite{pitassi2016poly}
that any depth $d$ Frege proof for a carefully constructed Tseitin
contradiction $\tau$ over 3-regular expanders must have size at least
$\exp(\Omega(\frac{\log n}{d})^{2})$. Improving the factor $(40pt2^{t})^{d}$
to $O(pt)^{d}$ in Theorem \ref{thm:prst} would imply exponential
lower bounds on the size of Frege proofs for $\tau$; more specifically,
it would imply that any depth $d$ Frege proof for $\tau$ over $3$-regular
expanders must have size at least $\exp(\Omega(n^{1/d}))$, matching
the optimal lower bounds known for boolean circuits \cite{hastad1986almost}. 

However, note that in the parameters for Theorem \ref{thm:prst},
if we let $B=U=\{0,1\}^{n}$, it shatters every set $A\in U$. Further,
we can choose $Y\in{\cal S}_{p,B}$ to simply be a random $p$-restriction
- this is equivalent to a distribution $W=V+u$, where the distribution
$V$ chooses a random subspace of $U$ by independently choosing the
$i$'th standard basis vector with probability $p$, and $u$ is a
uniformly random bit string in $U$. With these instantiations, the
following is an immediate corollary of Theorem \ref{thm:lower_bound}.
\begin{cor}
\label{thm:basic_lowerbound}For every $t$, there is an affine subspace
$B$ of $\{0,1\}^{n}$, and $g_{t}\in{\cal F}_{t,B}$ and $Y\in{\cal S}_{p,B}$,
such that for $0\leq p\leq c_{p}2^{-t}$ and $0\leq d\leq c_{d}\left(\frac{\log n}{2^{t}\log t}\right)$,
\[
\Pr_{\rho\leftarrow Y}[\text{DT}_{\text{depth}}(g_{t}|_{\rho})\geq d]\geq(c_{0}p2^{t})^{d},
\]
where $c_{p}$, $c_{d}$ and $c_{0}$ are universal constants. 
\end{cor}

Note that Theorem \ref{thm:basic_lowerbound} shows that the bounds
in Theorem \ref{thm:prst} of \cite{pitassi2016poly} are almost tight.
As a consequence, it shows that for small values of $d$, up till
$\sim2^{-t}\log n$, there are functions that have $t$-clipped decision
trees, but for which exponential lower bounds on the size of Frege
proofs are not possible via improvements to the switching lemma stated
in Theorem \ref{thm:prst}. Note however that it leaves open the possibility
of obtaining such exponential bounds for $t$-DNFs if they are treated
\emph{directly} and not cast as $t$-clipped decision trees. \\

We start by proving Theorem \ref{thm:upper_bound} in Section \ref{sec:upper_bound},
since observations from the improved upper bound proof, stated in
Subsection \ref{sub:Obs_uppbound}, will lead to the definition of
tree tribes in Section \ref{sec:Tree-tribes}, which will further
help in proving Theorem \ref{thm:lower_bound} in Section \ref{sec:Lower-Bounds}.

\section{Preliminaries}

\subsection{Decision trees and random restrictions}

We will consider only boolean functions over the hypercube, $f:\{0,1\}^{n}\rightarrow\{0,1\}$.
When we consider the function $f$ in the fourier basis, we assume
that it is a boolean function over $\{1,-1\}^{n}$. The fourier expansion
of $f:\{1,-1\}^{n}\rightarrow\{1,-1\}$ is given by $f(x)=\sum_{S\subseteq[n]}\hat{f}_{S}\chi_{S}(x)$,
where $\chi_{S}(x)=\prod_{i\in S}x_{i}$ and for every $S$, $\hat{f}_{S}\in\mathbb{R}$.
The bias of $f$ is defined as 
\[
\text{bias}(f)=\left|\Pr_{x}[f(x)=0]-\frac{1}{2}\right|,
\]
and the correlation between boolean functions $f$ and $g$ is given
by 
\[
\text{Corr}(f,g)=\Pr_{x}[f(x)=g(x)].
\]
The influence of a variable $x_{i}$ is defined as the probability,
over choosing a random bit string on variables different from $x_{i}$,
that flipping the value of $x_{i}$ flips the value of the function. 
\begin{defn}
\emph{\label{def:(Decision-trees)}(Decision trees)} Given a set of
variables $X=\{x_{1},\ldots,x_{n}\}$, a \emph{decision tree} $T$
is a rooted binary tree $T=(V,E,X,\sigma_{V},\sigma_{E})$, where
the functions $\sigma_{V}:V\rightarrow X$ and $\sigma_{E}:E\rightarrow\{0,1\}$
respectively label the vertices with variables in $X$ and edges with
$0$ or $1$, and the leaves are labelled with either $0$ or $1$.
Further, in every root to leaf path, any variable of $X$ appears
\emph{at most once}. Given a boolean function $f:\{0,1\}^{n}\rightarrow\{0,1\}$
over the variables $X$, we say that $T$ is a decision tree for $f$
or that $f$ is computed/evaluated by $T$, if, for every root to
leaf path $\pi=\{x_{i_{1}},e_{i_{1}},\ldots,x_{i_{l}},e_{i_{l}},b\}$
where $x_{i_{k}}\in X$, $e_{i_{k}}\in\{0,1\}$ and where the value
of the leaf is $b\in\{0,1\}$, the function $f$ evaluates to $b$
on the subcube in which the variables $(x_{i_{1}},\ldots,x_{i_{l}})$
are assigned the value $(e_{i_{1}},\ldots,e_{i_{l}})$. Given a decision
tree $T$, we use the symbols $f(T)$ to denote the function computed
by the tree on variables $X$. 
\begin{defn}
\emph{($\text{DT}_{\text{depth}}$)} For any root to leaf path $\pi=(x_{i_{1}},e_{i_{1}},\ldots,x_{i_{l}},e_{i_{l}},b)$,
we say that the \emph{length} of the path is the number of edges in
it, i.e. $|\pi|=l$. The \emph{depth} of a decision tree $T$, denoted
by $\text{depth}(T)$ is the length of the \emph{longest} root to
leaf path in $T$. For any boolean function $f$, we denote $\text{DT}_{\text{depth}}(f)$
to be the minimum depth amongst all decision trees $T$ for $f$.
We also say that the variables \emph{queried} along the path $\pi$
were $(x_{i_{1}},\ldots,x_{i_{l}})$, and the values \emph{assigned
}or \emph{received }or the values to which the variables \emph{evaluated}
were $(e_{i_{1}},\ldots,e_{i_{l}})$.
\begin{defn}
\emph{(Clipped decision trees)} A decision tree $T$ is $t$-clipped,
if every vertex of $T$ is at a distance of at most $t$ from some leaf. 
\begin{defn}
\emph{(Random restrictions)} Given a set of variables $X=\{x_{1},\ldots,x_{n}\}$,
a restriction $\rho$ is a string in $\{0,1,*\}^{n}$, i.e., $\rho:X\rightarrow\{0,1,*\}$,
where $\rho(x_{i})=*$ means that the variable is left unset, i.e.,
$\rho(x_{i})=x_{i}$. A random restriction is a distribution over
$\{0,1,*\}^{n}$. A restriction $\rho$ is said to be a random $p$-restriction,
if $\rho$ is a distribution that choses a subset $S\subseteq X$
of variables with probability $|S|^{p}$ and assigns $*$ to them,
and uniformly assigns the value $0$ or $1$ to the remaining variables.
Equivalently, $\rho$ independently assigns each variable the value
$*$ with probability $p$ and 0 or 1 with probability $q=\frac{1-p}{2}$.
\end{defn}
\end{defn}
\end{defn}
We will restrict the symbols $p$ and $q$ to that specific meaning
throughout, even when we treat them as formal variables.\end{defn}
\begin{rem}
Note that a random $p$-restrictions $\rho$ has product structure
due to the independence between various variables, i.e. $\rho=\rho_{1}\rho_{2}$,
where $\rho_{1}$ and $\rho_{2}$ are random $p$-restrictions over
the variable sets $\{x_{1},\ldots,x_{m}\}$ and $\{x_{m+1},\ldots,x_{n}\}$
respectively. 
\end{rem}
Given a boolean function $f$ on $n$ variables, we write $f|_{\rho}$
for the boolean function obtained by restricting $f$ according to
a subset of $\{0,1\}^{n}$ chosen according to $\rho$, and say that
$f$ is restricted to $\rho$ or $f$ is hit by $\rho$. We will also
use the notation $T|_{\rho}$, which would mean $f(T)|_{\rho}$. We
will mainly study the probability of the event that the $\text{DT}_{\text{depth}}$
of $f|_{\rho}$ is greater than $d$ when $\rho$ is a random $p$-restriction,
i.e., 
\[
\Pr_{\rho}[\text{DT}_{\text{depth}}(f|_{\rho})\geq d].
\]

\subsection{Polynomials and linear operators}

A univariate polynomial $Q$ in the variable $p$ will be an infinite
dimensional vector over $\mathbb{R}$, in the vector space $\mathbb{R}[p]$,
with a finite number of non-zero coefficients. We will denote $Q$
as $Q=\sum_{i\geq0}c_{i}p^{i}$. Using the standard notation for generating
functions, for every $i\geq0$, we define the operator $[p^{i}]$
as 
\[
[p^{i}]Q=c_{i}.
\]
Further, we define the operator $[\uparrow p^{i}]$ as

\[
[\uparrow p^{i}]Q=\sum_{j\geq i}\left([p^{j}]Q\right)p^{j-i}=\sum_{j\geq i}c_{j}p^{j-i}.
\]

\begin{defn}
\emph{(Absolute maximizer)} Given a closed set ${\cal D}\subseteq[0,1]$
and some polynomial $Q$ in $\mathbb{R}[p]$, we define the absolute
maximizers $G_{i}$ of  $[\uparrow p^{i}]Q$ as
\[
G_{i}(Q)=\max_{p\in{\cal D}}\left|[\uparrow p^{i}]Q\right|.
\]

\end{defn}
The following claim is immediate from the definitions. 
\begin{lem}
\label{claim:poly_props}For any two univariate polynomials $Q,R\in\mathbb{R}[p]$,
the following hold:
\begin{enumerate}
\item $[p^{i}](Q+R)=[p^{i}]Q+[p^{i}]R$
\item $[\uparrow p^{i}](Q+R)=[\uparrow p^{i}]Q+[\uparrow p^{i}]R$
\item $[p^{i}](QR)=\sum_{j=0}^{i}[p^{j}]Q[p^{i-j}]R$
\item $[p^{i}]p^{j}Q=0$ if $j>i$
\item $[p^{i}]p^{j}Q=[p^{i-j}]Q$ if $j\leq i$
\item $G_{i}(Q\pm R)\leq G_{i}(Q)+G_{i}(R).$
\end{enumerate}
\end{lem}

\subsection{Basic equalities and inequalities}

We mention some basic equalities and inequalities that we will repeatedly
use:
\begin{lem}
\label{claim:basic_ineqs}For any integer $n\geq1$, the following
hold:
\begin{enumerate}
\item 
\[
\sum_{i=0}^{n}p^{i}=\dfrac{1-p^{n+1}}{1-p}
\]

\item 
\[
\sum_{k=1}^{n}\frac{k}{2^{k}}=2-\frac{n+2}{2^{n}}\leq2
\]

\item 
\[
\sum_{k=2}^{n}{k \choose 2}\frac{1}{2^{k}}=2-\frac{n^{2}+3n+4}{2^{n+1}}\leq2
\]

\item 
\[
\sum_{i=0}^{n}{n \choose i}q^{n-i}p^{i}(i+1)=pn(q+p)^{n-1}+(q+p)^{n}
\]

\item For $p\in[0,1]$, 
\[
\frac{1}{p}\left(1-(1-p)^{t}\right)\leq t
\]

\item For $p\in[0,1]$,
\[
\frac{1}{p{}^{2}}\left((1-p)^{t}-1+tp\right)\leq{t \choose 2}
\]

\end{enumerate}
\end{lem}
\begin{proof}
(1), (2) and (3) follow by direct calculations. For (4), write $x(q+x)^{n}=\sum_{i=0}^{n}{n \choose i}q^{n-i}x^{i+1}$,
differentiate both sides with respect to $x$, and set $x=p$. We
show (6) and the proof of (5) is similar. Use induction on $t$ to
show that the function $\frac{1}{p^{2}}\left((1-p)^{t}-1+tp\right)$
is a non-increasing function of $p$. It is trivially true for $t=1$,
let it be true for $t-1$. Rewriting, 
\begin{eqnarray*}
\frac{1}{p^{2}}\left((1-p)^{t}-1+tp\right) & = & \frac{1}{p^{2}}\left((1-p)\left((1-p)^{t-1}-1+p(t-1)\right)+(1-p)-p(1-p)(t-1)-1+tp\right)\\
 & = & (1-p)\frac{1}{p^{2}}\left((1-p)^{t-1}-1+p(t-1)\right)+t-1
\end{eqnarray*}
which is a non-increasing function of $p$, since $(1-p)$ is a decreasing
function of $p$ in $[0,1]$ and by the induction hypothesis, $\frac{1}{p^{2}}\left((1-p)^{t-1}-1+p(t-1)\right)$
is a non-increasing function of $p$. Thus, the function is maximized
at $p=0$, setting which proves the claim. 
\end{proof}

\section{\label{sec:upper_bound}Upper bound}

We start by by showing an improved upper bound on the probability
that a function represented by a clipped decision tree has depth greater
than $d$ after it is hit with a random $p$-restriction. Our proof
will be recursive and use conditioning similar to that in \cite{hastad1986almost}. 

Let $T$ be a $t$-clipped decision tree and $f=f(T)$ be the corresponding
boolean function. We assume that $T$ has $n$ variables, but note
that since the final bounds in Theorem \ref{thm:upper_bound} are
independent of $n$, it is safe to keep the intuition that $T$ is
virtually an infinite tree. Let $T$ be hit by a random $p$-restriction
$\rho$. We want to show that 
\[
\Pr_{\rho}[\text{DT}_{\text{depth}}(f|_{\rho})\geq d]\leq(4p2^{t})^{d}.
\]
We will use induction on $d$ to prove the claim. To set up the induction,
we will need the following definitions. 
\begin{defn}
A decision tree $T$ is $(t_{0},t)$-clipped for $t_{0}\leq t$, if
the root has distance at most $t_{0}$ to some leaf, and every other
vertex has distance at most $t$ to some leaf. 
\begin{defn}
For integers $t_{0},t,n,d$, define the probabilities $\gamma_{d,n}(t_{0},t)$
as follows:
\[
\gamma_{d,n}(t_{0},t)=\max_{\text{ }\substack{(t_{0},t)\text{-clipped trees }T\text{ that}\\
\text{decide any function on \ensuremath{n} variables}
}
}\Pr_{\rho}[\text{DT}_{\text{depth}}(T|_{\rho})\geq d]
\]

\end{defn}
\end{defn}
The probabilties $\gamma$ have been specifically defined to make
the parameter $n$ irrelevant, as long as any recursive inequality
for $\gamma$ only reduces $n$. It is simple to see that $\gamma$
is monotone in $n$. 
\begin{lem}
\label{lem:gam_n_inc}If $n'\leq n$, then $\gamma_{d,n'}(t_{0},t)\leq\gamma_{d,n}(t_{0},t)$.\end{lem}
\begin{proof}
Assume $n=n'+1$, and the claim follows. Let $T'$ be a $(t_{0},t)$
clipped tree on $n'$ vertices, which achieves the maximum for $\gamma_{d,n'}(t_{0},t)$
for some $d$. Let $x$ be the last variable queried in $T'$ along
some root to leaf path, such that the subtrees rooted at $x$ are
both leaves. Note that such a variable always exists in a finite tree.
When $x$ is queried and it evaluates to 1, let the value of the leaf
be $a\in\{0,1\}$. Let $y$ be a new variable different from all the
variables appearing in $T'$. Let $T$ be a new tree created as follows:
$T$ is same as $T'$, except that when the variable $x$ is queried,
if it evaluates to 1, we query the variable $y$, and the value of
both the leaves at $y$ is $a$. Note that essentially, $y$ will
never be queried by any decision tree. $T$ is also $(t_{0},t)$-clipped,
and since $f(T')=f(T)$ for any $y\in\{0,1,*\}$, they have the same
minimum depth under the action of any random restriction, and the
claim follows. 
\end{proof}

\subsection{Recurrence for $\gamma$}

We proceed by writing a recurrence for $\gamma_{d,n}(t_{0},t)$. 

Let $T$ be the $(t_{0},t)$-clipped decision tree for which $\gamma_{d,n}(t_{0},t)$
has maximum value. Without loss of generality, let $x_{1}$ be the
root variable queried in $T$. Let $(x_{1},e_{1},\ldots,x_{t_{0}},e_{t_{0}},a)$
be the path from the root to the leaf at distance $t_{0}$ which evaluates
to $a\in\{0,1\}$. Since any variable is assigned $0$ or $1$ with
equal probability, without loss of generality, let $e_{1}=0$. Let
the subtree out of the 0-edge at $x_{1}$ be $T_{0}$ and the subtree
out of the 1-edge be $T_{1}$. 

Under the action of a random $p$-restriction $\rho$, if $x_{1}$
is assigned $0$ by $\rho$, note that we get a $(t_{0}-1,t)$ clipped
decision tree $T_{0}$ on $n_{0}$ variables where $n_{0}<n$. By
our definition of decision trees \ref{def:(Decision-trees)}, since
$x_{1}$ appears as the root of $T$, it cannot appear again as a
variable in $T_{0}$, and thus $T_{0}$ is indeed a $(t_{0}-1,t)$-clipped
tree. If $x_{1}$ is assigned $1$ by $\rho$, similarly, we get a
$(t,t)$ clipped decision tree $T_{1}$ on $n_{1}$ variables where
$n_{1}<n$. 

Let the event $E_{T,d}$ be defined as follows:
\[
E_{T,d}\equiv\text{DT}_{\text{depth}}(T|_{\rho})\geq d.
\]

\begin{lem}
\label{claim:Has_obs}In case $x_{1}$ is assigned $*$ by a random
$p$-restriction $\rho$, 
\[
E_{T,d}\subseteq E_{T_{0},d-1}\bigcup E_{T_{1},d-1}.
\]
\end{lem}
\begin{proof}
We show the contrapositive. Let $\rho'$ be the restriction on variables
different from $x_{1}$. As stated before, by definition \ref{def:(Decision-trees)},
$T_{0}$ and $T_{1}$ do not contain $x_{1}$, and thus $T_{0}|_{\rho}=T_{0}|_{\rho'}$
and $T_{1}|_{\rho}=T_{1}|_{\rho'}$. Let $V_{0}$ and $V_{1}$ be
decision trees such that $\text{depth}(V_{0})=\text{DT}_{\text{depth}}(T_{0}|_{\rho})$
and $\text{depth}(V_{1})=\text{DT}_{\text{depth}}(T_{1}|_{\rho})$
($V_{0}$ and $V_{1}$ are random variables). Thus, if $\text{DT}_{\text{depth}}(T_{0}|_{\rho})<d-1$,
i.e. the smallest depth of the decision tree computing $T_{0}|_{\rho}$
is strictly less than $d-1$, and similarly $\text{DT}_{\text{depth}}(T_{1}|_{\rho})<d-1$,
then the decision tree which has $x_{1}$ as the root, and $V_{0}$
as the left subtree and $V_{1}$ as the right subtree would correctly
evaluate the function $T|_{\rho}$ and have depth strictly less than
$d$, which would mean that $\text{DT}_{\text{depth}}(T|_{\rho})<d$,
and the event $E_{T,d}$ cannot happen. 
\end{proof}
Let $\rho=\rho_{x_{1}}\rho'$ where $\rho'$ is a random restriction
on variables different from $x_{1}$. Thus, we can write the following,
\begin{eqnarray}
\Pr_{\rho}[E_{T,d}] & = & \Pr_{\rho}[E_{T,d}|\rho(x_{1})=0]\Pr_{\rho}[\rho(x_{1})=0]+\Pr_{\rho}[E_{T,d}|\rho(x_{1})=1]\Pr_{\rho}[\rho(x_{1})=1]\nonumber \\
 &  & +\Pr_{\rho}[E_{T,d}|\rho(x_{1})=*]\Pr_{\rho}[\rho(x_{1})=*]\nonumber \\
 & \leq & \Pr_{\rho'}[E_{T_{0},d}]q+\Pr_{\rho'}[E_{T_{1},d}]q+\Pr_{\rho'}[E_{T_{0},d-1}\bigcup E_{T_{1},d-1}]p\nonumber \\
 & \leq & \Pr_{\rho'}[E_{T_{0},d}]q+\Pr_{\rho'}[E_{T_{1},d}]q+\Pr_{\rho'}[E_{T_{0},d-1}]p+\Pr_{\rho'}[E_{T_{1},d-1}]p\label{eq:prob_events_for_gamma}
\end{eqnarray}
where the second line used Lemma \ref{claim:Has_obs}, the fact that
the subtrees at $x_{1}$ do not contain the variable $x_{1}$ and
that $\rho$ has product structure since each of the variables are
assigned values independently, and the last line used the union bound.
Further, note that since $\Pr_{\rho}[E_{T,d}]\leq\gamma_{d,n}(t_{0},t)$
if $T$ is $(t_{0},t)$-clipped and has $n$ variables, we can rewrite
the inequality \ref{eq:prob_events_for_gamma} in terms of the $\gamma_{d,n}(t_{0},t)$,
and using Lemma \ref{lem:gam_n_inc}, we get, 
\begin{eqnarray*}
\gamma_{d,n}(t_{0},t) & \leq & q\gamma_{d,n_{0}}(t_{0}-1,t)+q\gamma_{d,n_{1}}(t,t)+p\gamma_{d-1,n_{0}}(t_{0}-1,t)+p\gamma_{d-1,n_{1}}(t,t)\\
 & \leq & q\gamma_{d,n}(t_{0}-1,t)+q\gamma_{d,n}(t,t)+p\gamma_{d-1,n}(t_{0}-1,t)+p\gamma_{d-1,n}(t,t).
\end{eqnarray*}
The parameters $n$ and $t$ can be made implicit, and we can rewrite
the recurrence succinctly as
\begin{eqnarray}
\gamma_{d}(t_{0}) & \leq & q\gamma_{d}(t_{0}-1)+q\gamma_{d}(t)+p\gamma_{d-1}(t_{0}-1)+p\gamma_{d-1}(t),\label{eq:gam_k_rec}
\end{eqnarray}
and for every integer $d$, we set 
\begin{equation}
\gamma_{d}(0)=0.\label{eq:gammad00}
\end{equation}
Notice that at this point, we can use induction for $\gamma_{d-1}(t)$
but not the other terms. We now show the following recursive claim
for $\gamma_{d}(t_{0})$. \\

\begin{lem}
After $m$ iterations, the recursion is,
\begin{equation}
\gamma_{d}(t_{0})\leq\sum_{i=0}^{m}{m \choose i}q^{m-i}p^{i}\gamma_{d-i}(t_{0}-m)+\sum_{i=1}^{m}q^{i}\gamma_{d}(t)+\sum_{j=1}^{m}p^{j}\gamma_{d-j}(t)\left(\sum_{i=0}^{m-j}{j+i \choose i}q^{i}\right).\label{eq:gam_m_rec}
\end{equation}
\end{lem}
\begin{proof}
The base case for $m=1$ is given by equation \ref{eq:gam_k_rec}.
Let the recurrence be true for $m$. Thus, 
\[
\gamma_{d}(t_{0})\leq\sum_{i=0}^{m}{m \choose i}q^{m-i}p^{i}\gamma_{d-i}(t_{0}-m)+\sum_{i=1}^{m}q^{i}\gamma_{d}(t)+\sum_{j=1}^{m}p^{j}\gamma_{d-j}(t)\left(\sum_{i=0}^{m-j}{j+i \choose i}q^{i}\right).
\]
Let $t_{0}'=t_{0}-m-1$. Using the recursion in equation \ref{eq:gam_k_rec}
for $\gamma_{d-i}(t_{0}-m)$, we can write 
\begin{eqnarray*}
\gamma_{d}(t_{0}) & \leq & \sum_{i=0}^{m}{m \choose i}q^{m-i}p^{i}\left(q\gamma_{d-i}(t_{0}')+q\gamma_{d-i}(t)+p\gamma_{d-i-1}(t_{0}')+p\gamma_{d-i-1}(t)\right)\\
 &  & +\sum_{i=1}^{m}q^{i}\gamma_{d}(t)+\sum_{j=1}^{m}p^{j}\gamma_{d-j}(t)\left(\sum_{i=0}^{m-j}{j+i \choose i}q^{i}\right).
\end{eqnarray*}
We can sum the terms based on whether the parameter in $\gamma$ is
$t_{0}'$ or $t$. Summing the first and the third terms above, we
can rewrite them as 
\begin{eqnarray*}
A_{1} & = & \sum_{i=0}^{m}{m \choose i}q^{m-i}p^{i}\left(q\gamma_{d-i}(t_{0}')+p\gamma_{d-i-1}(t_{0}')\right)\\
 & = & q^{m+1}\gamma_{d}(t_{0}')+p^{m+1}\gamma_{d-m-1}(t_{0}')+\sum_{i=1}^{m}\left({m \choose i}+{m \choose i-1}\right)q^{m-i+1}p^{i}\gamma_{d-i}(t_{0}')\\
 & = & \sum_{i=0}^{m+1}{m+1 \choose i}q^{m+1-i}p^{i}\gamma_{d-i}(t_{0}').
\end{eqnarray*}
And summing the remaining terms, we get, 
\begin{eqnarray*}
A_{2} & = & q^{m+1}\gamma_{d}(t)+\sum_{i=1}^{m}{m \choose i}q^{m-i+1}p^{i}\gamma_{d-i}(t)+\sum_{i=0}^{m}{m \choose i}q^{m-i}p^{i+1}\gamma_{d-i-1}(t)\\
 &  & +\sum_{i=1}^{m}q^{i}\gamma_{d}(t)+\sum_{j=1}^{m}p^{j}\gamma_{d-j}(t)\left(\sum_{i=0}^{m-j}{j+i \choose i}q^{i}\right)\\
 & = & \sum_{i=1}^{m+1}q^{i}\gamma_{d}(t)+p^{m+1}\gamma_{d-m-1}(t)+\sum_{j=1}^{m}p^{j}\gamma_{d-j}(t)\left({m+1 \choose j}q^{m-j+1}+\sum_{i=0}^{m-j}{j+i \choose i}q^{i}\right)\\
 & = & \sum_{i=1}^{m+1}q^{i}\gamma_{d}(t)+\sum_{j=1}^{m+1}p^{j}\gamma_{d-j}(t)\left(\sum_{i=0}^{m+1-j}{j+i \choose i}q^{i}\right).
\end{eqnarray*}
Taking $A_{1}+A_{2}$ gives the inequality for the $(m+1)$'th iteration,
and proves the claim. 
\end{proof}

\subsection{Upper bound on $\gamma$}

Setting $t_{0}=t$ and $m=t$ in equation \ref{eq:gam_m_rec} and
using \ref{eq:gammad00}, we get, 
\begin{eqnarray*}
\frac{1-2q+q^{t+1}}{1-q}\gamma_{d}(t) & \leq & \sum_{j=1}^{t}p^{j}\gamma_{d-j}(t)\left(\sum_{i=0}^{t-j}{j+i \choose i}q^{i}\right).
\end{eqnarray*}
Using the induction hypothesis, for $\mu=\kappa2^{t}$, we have that
$\gamma_{d-c}(t)\leq(\mu p)^{d-c}$ for all $p,c,d,t$ where $p\leq\frac{1}{\mu}$.
Note that $\gamma$ is a probability and since $p\leq\frac{1}{\mu},$
it is also valid when $c>d$. \footnote{In fact, whenever $d<t$, we can indeed get much fewer terms in the
summation, and get a better bound, although asymptotically it does
not make a difference. } Then we have,
\begin{eqnarray}
\frac{1-2q+q^{t+1}}{1-q}\gamma_{d}(t) & \leq & (\mu p)^{d}\left(\sum_{j=1}^{t}\sum_{i=0}^{t-j}\left(\frac{1}{\mu}\right)^{j}{j+i \choose i}q^{i}\right).\label{eq:gam_bef_solving}
\end{eqnarray}
To get that $\gamma_{d}(t)\leq(\mu p)^{d}$, we only need to show
the following lemma. 
\begin{lem}
For $t\geq1$ and $\mu=4\cdot2^{t}$, 
\[
\left(\sum_{j=1}^{t}\sum_{i=0}^{t-j}\left(\frac{1}{\mu}\right)^{j}{j+i \choose i}q^{i}\right)\left(\frac{1-q}{1-2q+q^{t+1}}\right)\leq1.
\]
\end{lem}
\begin{proof}
The summations in the first term stated above can be made simpler
by adding diagonally, i.e., for values for which $j+i=l$. Note that
$1\leq j+i\leq t$. Thus, setting $i=l-j$, and using the fact that
${j+i \choose i}={j+i \choose j}$, 
\begin{eqnarray}
\sum_{j=1}^{t}\sum_{i=0}^{t-j}\left(\frac{1}{\mu}\right)^{j}{j+i \choose i}q^{i} & = & \sum_{j=1}^{t}\sum_{l=j}^{t}\left(\frac{1}{\mu}\right)^{j}{l \choose j}q^{l-j}\nonumber \\
 & = & \sum_{l=1}^{t}\sum_{j=1}^{l}\left(\frac{1}{\mu}\right)^{j}{l \choose j}q^{l-j}\nonumber \\
 & = & \sum_{l=1}^{t}\left(\frac{1}{\mu}+q\right)^{l}-\sum_{l=1}^{t}q^{l}\label{eq:mid_sum_1-1}\\
 & = & r\frac{1-r^{t}}{1-r}-q\frac{1-q^{t}}{1-q}\nonumber 
\end{eqnarray}
where 
\begin{equation}
r=q+\frac{1}{\mu}.\label{eq:r_eg}
\end{equation}
Let 
\begin{eqnarray}
U & = & \left(r\frac{1-r^{t}}{1-r}-q\frac{1-q^{t}}{1-q}\right)\frac{1-q}{1-2q+q^{t+1}}.\label{eq:U_eq}
\end{eqnarray}
We first show that for $t\geq1$, $\frac{dU}{dp}<0.$ From equation
\ref{eq:mid_sum_1-1}, let 
\[
A=\sum_{l=1}^{t}\left(\frac{1}{\mu}+q\right)^{l}-\sum_{l=1}^{t}q^{l}
\]
and 
\[
B=\frac{1-q}{1-2q+q^{t+1}}=\frac{1+p}{2p+\frac{(1-p)^{t+1}}{2^{t}}}=\frac{C}{D}.
\]
Since $\frac{1}{\mu}+q>q$, we have that $A>0$. Since $0\leq p\leq1$,
we have that $B>0$ and $D>0$.

\[
\frac{dA}{dp}=\frac{1}{2}\left(\sum_{l=1}^{t}l\left(q^{l-1}-\left(\frac{1}{\mu}+q\right)^{l-1}\right)\right)<0
\]
since each of the terms inside the summation is strictly less than
0. Further, 
\begin{eqnarray*}
\frac{dB}{dp} & = & \frac{1}{D^{2}}\left(2p+\frac{(1-p)^{t+1}}{2^{t}}-(1+p)\left(2-\frac{t+1}{2^{t}}(1-p)^{t}\right)\right)\\
 & = & \frac{1}{D^{2}}\left(2p+\frac{(1-p)^{t+1}}{2^{t}}-2-2p+\frac{t+1}{2^{t}}(1-p)^{t}(1+p)\right)\\
 & = & \frac{1}{D^{2}}\left(\frac{(1-p)^{t+1}}{2^{t}}+\frac{t+1}{2^{t}}(1-p)^{t-1}(1-p^{2})-2\right)
\end{eqnarray*}
where in the last line, we used the fact that $t\geq1$. Note that
since the terms involving $p$ are all decreasing in $p$, they are
maximized for $p=0$. Thus, in the last line
\begin{eqnarray*}
\frac{dB}{dp} & \leq & \frac{1}{D^{2}}\left(\frac{t+2}{2^{t}}-2\right)<0,
\end{eqnarray*}
where the last inequality used the fact that $t\geq1$. Thus, we have
that 
\[
\frac{dU}{dp}=A\frac{dB}{dp}+B\frac{dA}{dp}<0.
\]
Since $U$ is a decreasing function of $p$, setting $p=0$, and $\mu=\kappa2^{t}$
in equation \ref{eq:r_eg}, we get, 
\[
r=\frac{1}{2}\left(1+\frac{2}{\kappa2^{t}}\right),
\]
and substituting $p=0$ in \ref{eq:U_eq}, we require 

\[
U=2^{t}\left(r\frac{1-r^{t}}{1-r}-1+2^{-t}\right)\leq1
\]
or

\[
2r-r^{t+1}\leq1
\]
or 
\[
\frac{4}{\kappa}\leq\left(1+\frac{2}{\kappa2^{t}}\right)^{t+1}
\]
which is implied for all $t\geq1$ by setting $\kappa=4$, and we
get the required bound. 
\end{proof}
We have thus shown that $\gamma_{d}(t)\leq(4p2^{t})^{d}$, and Theorem
\ref{thm:upper_bound}, which we reproduce here for convenience. 
\begin{thm}
For any boolean function $f$ that has a $t$-clipped decision tree,
for a random $p$-restriction $\rho$, 
\[
\Pr_{\rho}[\text{DT}_{\text{depth}}(f|_{\rho})\geq d]\leq(4p2^{t})^{d}.
\]

\end{thm}

\subsection{\label{sub:Obs_uppbound}Observations from the upper bound}

We make a few observations from the upper bound proof that will help
us create structures for which we could prove a lower bound. In the
recurrence \ref{eq:prob_events_for_gamma} of total probability, since
we used the union bound in our proof of Theorem \ref{thm:upper_bound},
we neglected the term 
\[
\Pr_{\rho'}[E_{T_{0},d-1}\bigcap E_{T_{1},d-1}].
\]
How worse can this term be? First we make the following observation:
Without loss of generality, let 
\[
\phi=x_{1}x_{2}\ldots x_{t}\vee\phi'
\]
be a $t$-DNF where $\phi'$ is also a $t$-DNF. If we write $\phi$
as a $t$-clipped decision tree, the variables $x_{1}$ to $x_{t}$
will be connected by 0-edges and end in a leaf. Further, the 1-edges
out of each of the variables $x_{i}$ will be connected to a $t$-clipped
decision $T_{i}$. But note that each of the $T_{i}$'s is a decision
tree for $\phi'$. However, each $T_{i}$ could have a different value
of $\text{DT}_{\text{depth}}(T_{i})$, due to the different variables
to which they are connected. More specifically, for instance, the
tree $T_{2}$ does not \emph{know} the value of $x_{3}$, but $T_{4}$
\emph{knows} that that $x_{3}=0$. However, since we will consider
the \emph{maximum }of $\text{DT}_{\text{depth}}(T_{i})$ over all
$i$, if we take $t$ to be much smaller than $n$, it would be safe
to assume that for every $i$, the value of $\text{DT}_{\text{depth}}(T_{i})$
is approximately the same. As a result, we would have 
\[
\Pr_{\rho'}[E_{T_{0},d-1}|E_{T_{1},d-1}]\approx1.
\]
Thus, the loss due to the union bound in this case is significant,
and if we want the union bound to be tight, we would essentially want
that 
\[
\Pr_{\rho'}[E_{T_{0},d-1}\bigcap E_{T_{1},d-1}]\approx0.
\]
However, this would require creating intricate (anti)-correlations
between variables, and it is unclear how that can be done. Instead,
it is possible to get, 
\[
\Pr_{\rho'}[E_{T_{0},d-1}\bigcap E_{T_{1},d-1}]\approx\Pr_{\rho'}[E_{T_{0},d-1}]\Pr_{\rho'}[E_{T_{1},d-1}],
\]
by ensuring that there are no correlations between the two events,
which would be true, at least if the variables in the two subtrees
$T_{0}$ and $T_{1}$ are different. If we had such a case, the one-step
recurrence \ref{eq:gam_k_rec} for $\gamma$ would become
\[
\gamma_{d}(t_{0})\approx q\gamma_{d}(t_{0}-1)+q\gamma_{d}(t)+p\gamma_{d-1}(t_{0}-1)+p\gamma_{d-1}(t)-p\gamma_{d-1}(t_{0}-1)\gamma_{d-1}(t).
\]
Further, since a $(t_{0}-1,t)$ clipped tree is also a $(t,t)$ clipped
tree for $t_{0}\leq t$, we have that 
\[
\gamma_{d-1}(t_{0}-1)\leq\gamma_{d-1}(t),
\]
and we can write 
\[
\gamma_{d}(t_{0})\apprge q\gamma_{d}(t_{0}-1)+q\gamma_{d}(t)+p\gamma_{d-1}(t_{0}-1)+p\gamma_{d-1}(t)-p\gamma_{d-1}^{2}(t).
\]
But note that we expect the variables $\gamma_{d-1}$ to exponentially
decrease with increase in $d$ in the final bound, and thus 
\[
\gamma_{d-1}^{2}(t)\ll\gamma_{d-1}(t),
\]
and we should have 
\[
\gamma_{d}(t_{0})\approx q\gamma_{d}(t_{0}-1)+q\gamma_{d}(t)+p\gamma_{d-1}(t_{0}-1)+p\gamma_{d-1}(t).
\]
This intuition turns out to be correct, as we show in section \ref{sec:Lower-Bounds},
for functions that we define next. 
\begin{rem}
The reason why we do not get a bound as strong as that for $t$-DNFs,
i.e. $O(pt)^{d}$, in Theorem \ref{thm:upper_bound} is that $t$-DNFs
have \emph{additional }structure which $t$-clipped decision trees
do not. Note that in setting up the recursion in equation \ref{eq:gam_k_rec},
we used a union bound for \emph{every} variable that was assigned
$*$ by $\rho$. However, in the case of $t$-DNFs, if we try to write
the same recursive expressions for every clause, if any variable in
the clause is assigned the value $0$ by $\rho$, the clause evaluates
to 0 \emph{even if} some other variables of the clause are assigned
$*$ by $\rho$, and a union bound would \emph{not }be necessary in
that case. Only if every variable in the clause is assigned $1$ or
$*$ by $\rho$ do we need to use a union bound for the variables
in a clause, a luxury that we do not have in the case of $t$-clipped
decision trees. 
\end{rem}

\section{\label{sec:Tree-tribes}Tree tribes}

\subsection{Tree tribes}

We now formally define tree tribes and their variants. A specific
variant, $t$-clipped xor tree tribe, will be the function that will
help to achieve the bounds stated in Theorem \ref{thm:lower_bound}. 
\begin{defn}
\emph{(Tree tribe)} A boolean function $f$ is called a tree tribe,
denoted by $\Xi$, if there is a decision tree $T$ deciding $f$,
such that all the variables at all the vertices of $T$ are distinct,
and from every vertex of $T$, there is a path to a leaf labelled
0 and a path to a leaf labelled 1. 
\end{defn}
An example for a tree tribe is the $OR$ function. 
\begin{rem}
In the definition of $\Xi$, the condition of having a path to a 0-leaf
and a path to a 1-leaf only ensures that no vertex is redundant, since
if all the paths from some vertex $x$ go to a 0-leaf (or 1-leaf),
$x$ can be replaced by a 0-leaf (or 1-leaf).
\end{rem}
Given definition of a tree tribe, it is possible to derive a variety
of different tree tribes by imposing additional structure, and we
define the specific structure that we'll need. 
\begin{defn}
\emph{(Complete clipped decision trees)} Denote a complete $t$-clipped
decision tree $T$ on $r$ levels by $W_{t}(r)$, and define it recursively
as follows: $W_{t}(0)$ is a leaf. $W_{t}(r)$ consists of vertices
$\{v_{1},\ldots,v_{t}\}$, a leaf denoted by $v_{t+1}$, and edges
$e_{i,0}$ and $e_{i,1}$ for $i\in\{1,\ldots,t\}$. The vertices
$v_{1}$ to $v_{t}$ will be said to belong to layer or level 1. Each
edge $e_{i,0}$ is labelled 0 and it will be called a $0$-edge, and
it connects $v_{i}$ and $v_{i+1}$. Each edge $e_{i,1}$ is labelled
1 and it will be called a $1$-edge, and it connects $v_{i}$ to the
root of $W_{t}(r-1)$.
\begin{defn}
\emph{(Clipped xor tree tribe)} A boolean function $f$ is called
a $t$-clipped xor tree tribe on $r$ levels, denoted by $\Xi_{t}(r)$,
if $T(f)$ is a tree tribe, and can be expressed as a complete $t$-clipped
decision tree on $r$ levels, in which the leaves are labelled by
the parity of the edges on the path from the root to the leaf.
\begin{defn}
\emph{(Level of variable $x$ in $\Xi_{t}(r)$)} We can define the
level formally, but the informal definition is cleaner. We define
the level of some variable $x$ in $\Xi_{t}(r)$ by the \emph{recursive}
step at which it was added to $\Xi_{t}(r)$. The variables at level
1 are $x_{1}$ to $x_{t}$, each of which is connected to a copy of
$\neg\Xi_{t}(r-1)$ on distinct variables. The first $t$ variables
in each of the $t$ independent copies of$\neg\Xi_{t}(r-1)$, a total
of $t^{2}$ variables, are at level 2, and each of them is connected
to a copy of $\Xi_{t}(r-2)$, and so on. 
\end{defn}
\end{defn}
\end{defn}
\noindent An example for $\Xi_{2}(3)$ is given in figure \ref{fig1}. The functions $\Xi_{t}(r)$ exhibit many nice properties, which we
discuss next.

\begin{figure}
	\centering
	\def\svgwidth{300bp}
	\def\svgscale{1}
	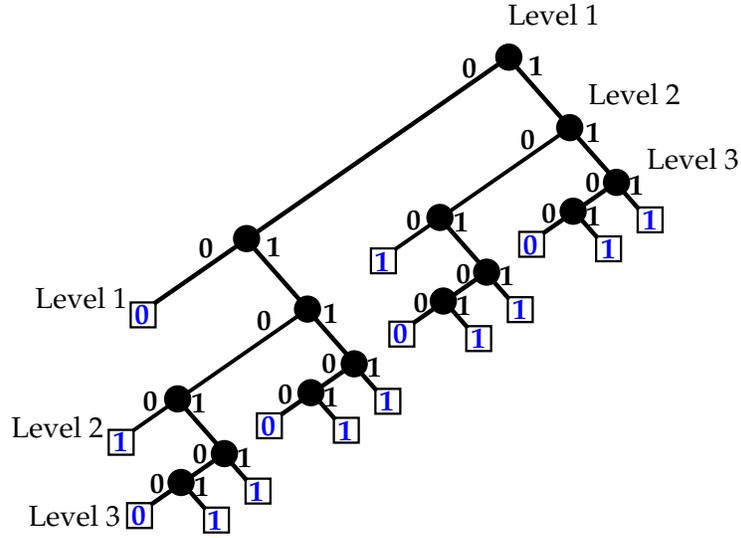
	\caption{A 2-clipped xor tree tribe on 3 levels, or $\Xi_2(3)$. The variables at each of the vertices are distinct, and the leaves are labelled by the parity of the edges along the root to leaf path. Note that there are 2 vertices at level 1, 4 vertices at level 2, and 8 vertices at level 3.}
	\label{fig1}
\end{figure}

\subsection{Properties of tree tribes}

We discuss some properties and observations about $\Xi$ and $\Xi_{t}(r)$.
Varying the parameters $t$ and $r$ lead to many interesting properties.

\subsubsection*{Basic properties}

The following properties of $\Xi_{t}(r)$ follow directly from the
definitions. 
\begin{enumerate}
\item The number of variables $n$ in $\Xi_{t}(r)$ is 
\[
n=\sum_{i=1}^{r}t^{i}=t\frac{t^{r}-1}{t-1}\approx t^{r}
\]
and if $t=1$, then $n=r$. 
\item All the leaves at the same level in $\Xi_{t}(r)$ have the same value,
i.e., the odd levels have leaves evaluating to 0 and the even levels
have leaves evaluating to 1. 
\item $\Xi_{t}(r)$ can be written as a DNF of width $O\left(\frac{t}{\log t}\log n\right)$. 
\item $\Xi_{t}(r)$ preserves its substructure. If $\neg\Xi_{t}(r)$ is
the negation of the function $\Xi_{t}(r)$, then all the 1-edges out
of the vertices in the 1st level of $\Xi_{t}(r)$ are connected to
$\neg\Xi_{t}(r-1)$. This is exactly the manner in which it behaves
like parity or xor of variables. 
\item \label{enu:t_delayed_parity}Every vertex in $\Xi_{t}(r)$ is at a
distance of at most $t+1$ from a leaf that evaluates to 0, \emph{and}
at a distance of at most $t+1$ from a leaf that evaluates to 1. 
\end{enumerate}
Amongst all these, we feel that observation (5)  is the main reason
that $\Xi_{t}(r)$ turns out to be resilient towards random $p$-restrictions.
Due to (5), $\Xi_{t}(r)$ behaves like a \emph{delayed }parity, or
more specifically a $t$-\emph{delayed} parity, since the value of
a vertex can affect the function after $(t+1)$ steps, and cause it
to evaluate to any of $0$ or $1$. However, it is different from
parity since a variable can affect the value of the function \emph{only
}if an input string takes a path from the root to the leaf through
that variable in the decision tree. This fact puts a limit on the
depth to which we can prove Theorem \ref{thm:lower_bound}.

\subsubsection*{Bias}

The bias of $\Xi_{t}(r)$, as is immediate from equations \ref{eq:exact_const_p0_odd}
and \ref{eq:exact_const_p0_even} that will be derived later, is given
by
\[
\text{bias}(\Xi_{t}(r))=\left|\frac{1-\left(\dfrac{1}{2^{t}}-1\right)^{r+1}}{2-\dfrac{1}{2^{t}}}-\frac{1}{2}\right|.
\]
For the case $r\gg t$, the bias is about $\frac{1}{2-2^{-t}}-\frac{1}{2}$.
If $t=1$, the bias is about $\frac{1}{6}$. As $t$ becomes large
enough, the bias goes to 0. For the case $r\ll t$, the bias is about
$\frac{1}{2}$. Note that if $r=1$, then $\Xi_{t}(1)$ is simply
the $OR$ function on $t$ variables. If $r\approx t$, then the bias
is about $\frac{1}{2}-\frac{t}{2^{t}}$. If $r\approx2^{t}$, then
the bias is about a small constant less than $\frac{1}{2}$.

\subsubsection*{Correlation with parity }

As observed in \ref{enu:t_delayed_parity}, the function $\Xi_{t}(r)$
behaves like a $t$-delayed parity in a certain sense. However, if
$t\geq2$, it has correlation exactly $\frac{1}{2}$ with parity,
i.e., same as that of a constant function. To see this, let $\overline{x}$
be a bit string on which $T(\overline{x})$ evaluates to 0, taking
some path from the root $x_{1}$ to some leaf $l$. Since all the
variables of $\Xi_{t}(r)$ are different and $t\geq2$, there is some
variable $y$ that is not on the path from $x_{1}$ to $l$, and flipping
its value cannot change the value of $T(\overline{x})$, but flipping
$y$ will always change the parity of the variables. Thus the correlation
of $\Xi_{t}(r)$ with parity is exactly $\frac{1}{2}$, when $t\geq2$.
When $t=1$, $\Xi_{1}(r)$ has exactly $r$ variables, and the correlation
with parity is about $\frac{1}{2}+\frac{1}{2^{r}}$ which again goes
to $\frac{1}{2}$ for large enough $r$.

\subsubsection*{Influence of variables}

An interesting thing about a tree tribe $\Xi$ in general is the influence
of variables. Let a variable $y$ be at a distance $d$ from the root.
Then note that $y$ has no influence over the value of the function
for any string $\overline{x}$ that takes a path in $\Xi$ from the
root to a leaf not passing through $y$. Thus, 
\[
\text{Inf}(y)\leq2^{-d}.
\]
As a result, the variables in a tree tribe have exponentially decreasing
influences with distance from the root.\footnote{The name \emph{tree dictators} seems more suitable due to this property,
but we call them tree tribes for succinctness.}

\subsubsection*{Fourier coefficients}

The fourier coefficients of $\Xi_{t}(r)$ exhibit interesting combinatorics.
\\

\noindent \textbf{The case $t=1$:} For $t=1$, we compute the fourier
coefficients recursively. Let the variables of $\Xi_{1}(n)$ be $X=(x_{1},\ldots,x_{n})$
taken naturally from the root $x_{1}$. We claim the following.
\begin{lem}
\label{lem:four_coeff_t_1}Let $S\subseteq X$ be a subset of $X=\{x_{1},\ldots,x_{n}\}$,
let the variables of $f=\Xi_{1}(n)$ be ordered naturally with $x_{1}$
as the root, and let $n$ be some odd number. Let $j\in\{1,\ldots,n\}$
be the largest index of a variable appearing in $S$. Then 
\[
\hat{f}_{S}=\frac{(-1)^{|S|+j}}{2^{n-1}}J_{n-j+1}
\]
where $J_{i}$ is the $i$'th Jacobsthal number, given by
\[
J_{i}=\frac{2^{i}-(-1)^{i}}{3}.
\]
\end{lem}
\begin{proof}
We assume the variables and leaves take values in $1$ or $-1$ without
loss of generality, corresponding to $0$ to $1$. If we write the
multilinear expansion of $f$, taking products of the terms $\frac{1\pm x_{i}}{2}$
from every root to leaf path, the term $\chi_{S}$ will appear only
for the leaves that appear after the vertex $x_{j}$. Further, since
$j$ is the largest index in $S$, the terms after $j$ will contribute
only the value $\pm\frac{1}{2}$. Thus we get that, 
\begin{eqnarray*}
\hat{f}_{S} & = & \frac{(-1)^{|S|-1}}{2^{j-1}}(-1)^{j+1}\left(\frac{1}{2}+\frac{1}{4}\sum_{i=0}^{n-j-2}\left(-\frac{1}{2}\right)^{i}\right)\\
 & = & \frac{(-1)^{|S|-1}}{2^{j-1}}(-1)^{j+1}\left(\frac{1}{2}+\sum_{i=2}^{n-j}\left(-\frac{1}{2}\right)^{i}\right)\\
 & = & \frac{(-1)^{|S|-1}}{2^{j-1}}(-1)^{j+1}\frac{2}{3}\left(1-\left(-\frac{1}{2}\right)^{n-j+1}\right)\\
 & = & \frac{(-1)^{|S|+j}}{2^{n-1}}J_{n-j+1}
\end{eqnarray*}
as required. Further properties of these numbers can be found at \cite{jacobsthalseq}.
\end{proof}
\noindent \textbf{The case for any general $t$.} 
\begin{lem}
\label{lem:four_coeff}Let $S$ be a subset of variables in $f=\Xi_{t}(r)$.
Let $y\in S$ be the variable which is farthest from the root of $\Xi_{t}(r)$.
Let $y$ be at a a distance $d$ from the root, at a distance $k$
from the closest leaf, and at level $l$ of $\Xi_{t}(r)$. Let $\alpha=2^{-t}$
and $\beta=2^{-t}-1$.

If the variables in $S$ do not lie on a path in $\Xi_{t}(r)$, then
$\hat{f}_{S}=0$. If all the variables of $S$ lie on a path, then
\[
\hat{f}_{S}=\pm\frac{1}{2^{k+d-1}}\left(\frac{1-\beta^{r-l+1}}{2-\alpha}\right).
\]
\end{lem}
\begin{proof}
We make two observations about $\hat{f}_{S}$. 
\begin{enumerate}
\item If the variables in $S$ do not all belong to the same path, i.e.,
for any two variables $x$ and $y$, if their least common ancestor
in $\Xi_{t}(r)$ is not one of $x$ or $y$, then $\hat{f}_{S}=0$.
This follows because if we write the multilinear expansion of $f$
for every root to leaf path, the term $\chi_{S}$ would never appear,
since all the variables in $\Xi_{t}(r)$ are different and there is
no path that contains both $x$ and $y$. 
\item The second observation is similar to the observation made in Lemma
\ref{lem:four_coeff_t_1}. Given observation (1), assume that all
the variables of $S$ lie on some path. If we write the multilinear
expansion for $f$, any path not passing through $y$ will not contain
the term $\frac{1\pm y}{2}$, and contribute to $0$ to the coefficient
of $\chi_{S}$. Further, since $y$ is the most distant variable from
the root, the variables after $y$ will contribute only the constant
values $\pm\frac{1}{2}$, and $y$ is connected to the function $\Xi_{t}(r-l+1)$
where the root has distance $k$ to a leaf out of its 0 edge, and
to the function $\neg\Xi_{t}(r-l)$ out of its 0 edge. 
\end{enumerate}
Let the parameter $t$ be fixed. Let $T$ be exactly $\Xi_{t}(r)$,
except that the root is at a distance $k$ from the leaf, and all
other vertices are at a distance $t$ from some leaf. If there is
probability $\frac{1}{2}$ of choosing any edge out of any vertex
of $T$, let $\mu_{r}(k)$ be the difference between the probability
of reaching an edge labelled 1 and an edge labelled -1. Then the recurrence
for $\mu$ is as follows:

\begin{eqnarray*}
\mu_{k}(r) & = & \frac{1}{2}\mu_{k-1}(r)-\frac{1}{2}\mu_{t}(r-1)\\
\mu_{0}(r) & = & 1.
\end{eqnarray*}
Solving the recurrence, we get,
\begin{eqnarray*}
\mu_{k}(r) & = & \frac{1}{2^{k}}-\left(\sum_{i=1}^{k}\frac{1}{2^{i}}\right)\mu_{t}(r-1)
\end{eqnarray*}
and evaluating for $k=t$, we get, 
\begin{eqnarray*}
\mu_{t}(r) & = & \frac{1}{2^{t}}-\left(\sum_{i=1}^{t}\frac{1}{2^{i}}\right)\mu_{t}(r-1)\\
 & = & \alpha+\beta\mu_{t}(r-1)\\
 & = & \alpha\sum_{i=0}^{r-1}\beta^{i}+\beta^{r}\\
 & = & \alpha\frac{1-\beta^{r}}{1-\beta}+\beta^{r}.
\end{eqnarray*}
Substituting the value for $\mu_{k}(r)$, we get 

\[
\mu_{k}(r)=\frac{1}{2^{k}}+\left(\frac{1}{2^{k}}-1\right)\left(\alpha\frac{1-\beta^{r-1}}{1-\beta}+\beta^{r-1}\right).
\]
Let us compute $\hat{f}_{S}$. Let $S'=S\backslash y$. In the multilinear
expansion of $f$, summing the terms containing coefficient of $\chi_{S}$
for all the paths passing through $y$, we get,
\[
\pm\frac{\chi_{S'}}{2^{d}}\left(\frac{1+y}{2}\mu_{k-1}(r-l+1)-\frac{1-y}{2}\mu_{t}(r-l)\right)=\pm\frac{\chi_{S}}{2^{d+1}}\left(\mu_{k-1}(r-l+1)+\mu_{t}(r-l)\right)
\]
and thus
\begin{eqnarray*}
\hat{f}_{S} & = & \pm\frac{1}{2^{d+1}}\left(\frac{1}{2^{k-1}}+\left(\frac{1}{2^{k-1}}-1\right)\mu_{t}(r-l)+\mu_{t}(r-l)\right)\\
 & = & \pm\frac{1}{2^{k+d}}\left(1+\mu_{t}(r-l)\right)
\end{eqnarray*}
and replacing the value of $\mu_{t}(r-l)$, it becomes
\[
\hat{f}_{S}=\pm\frac{1}{2^{k+d-1}}\left(\frac{1-\beta^{r-l+1}}{2-\alpha}\right).
\]

\end{proof}
Note that the magnitude of the fourier coefficients in Lemma \ref{lem:four_coeff}
depends only on the parameters $k$, $l$, and $d$ for any set $S$.
In particular, the parameters for the farthest variable from the root
in $S$ completely determines the \emph{magnitude} of $\hat{f}_{S}$,
and other variables of $S$ only affect the \emph{sign} of $\hat{f}_{S}$.
Further, if $r\gg t$, then for paths or sets $S$ in which the farthest
vertex is close to the root, $\beta\approx0$, and 
\[
|\hat{f}_{S}|\approx\frac{1}{2^{k+d-1}}\left(\frac{1}{2-2^{-t}}\right),
\]
implying that all the fourier coefficients are exponentially decreasing
in $k$ and $d$. If we consider $r$ and $t$ as parameters, we get
many different sequences of numbers based on the manner in which the
parameters - $r$, $t$, $k$, $l$, $d$ - are set, and these sequences
might demonstrate many interesting properties of the function, but
we do not explore that here.

\section{\label{sec:Lower-Bounds}Lower bound}

We now prove Theorem \ref{thm:lower_bound} for $t$-clipped xor tree
tribes or $\Xi_{t}(r)$, by induction on $d$. However, we cannot
use $d=0$ as the base case, and we discuss this when we do the inductive
step. The base case will be $d=1$, which we solve next. We do not
optimize any constants in this section.

\subsection{The case for $d=1$}

The base case $d=1$ turns out to be the most interesting. Here, we
want to evaluate the probability that the function $\Xi_{t}(r)$ has
depth more than 1 after being hit with a random $p$-restriction.
More specifically, we want to show the following.

\[
\Pr_{\rho}[\text{DT}_{\text{depth}}(\Xi_{t}(r)|_{\rho})\geq1]\geq c_{0}p2^{t}.
\]
Note that this cannot be true for small $r$, since if the function
has fewer than $\sim2^{t}$ variables, then the number of variables
assigned $*$ will be low on average, and the event $\text{DT}_{\text{depth}}(\Xi_{t}(r)|_{\rho})\geq1$
would be extremely unlikely. Thus, we would show the following.
\begin{lem}
\label{lem:d_1_lowerbound}For some universal constants $c_{0}$ and
$c_{p}$, for $r\in\Omega(2^{t})$ and $0\leq p\leq c_{p}2^{-t}$,
if $\rho$ is a random $p$-restriction, then
\[
\Pr_{\rho}[\text{DT}_{\text{depth}}(\Xi_{t}(r)|_{\rho})\geq1]\geq c_{0}p2^{t}.
\]

\end{lem}
Note that in Lemma \ref{lem:d_1_lowerbound}, we require that the
\emph{smallest }depth decision tree for $\Xi_{t}(r)|_{\rho}$ has
depth greater than 1 with good probability, which means that \emph{any}
decision tree representing $\Xi_{t}(r)|_{\rho}$ must query at least
one variable. This will be made possible by a simple observation. 
\begin{lem}
\label{claim:obs_path_to_0_1}$\text{DT}_{\text{depth}}(\Xi_{t}(r)|_{\rho})\geq1$
if and only there is a path in $\Xi_{t}(r)|_{\rho}$ from the root
to a leaf that evaluates to 0 and to a leaf that evaluates to 1.\end{lem}
\begin{proof}
Let $\text{DT}_{\text{depth}}(\Xi_{t}(r)|_{\rho})\geq1$. If all the
paths from the root of $\Xi_{t}(r)|_{\rho}$ were to a 0-leaf (or
1-leaf), then the function is the 0 function (or 1 function), which
has $\text{DT}_{\text{depth}}(\Xi_{t}(r)|_{\rho})=0$, a contradiction.
Now assume that there is a path from the root of $\Xi_{t}(r)|_{\rho}$
to a 0-leaf and a path to a 1-leaf. This means that in the truth-table
of $\Xi_{t}(r)|_{\rho}$, there is a string for which the function
is 0 and a string for which the function is 1, implying that it is
not equivalently the 0 or 1 function, and $\text{DT}_{\text{depth}}(\Xi_{t}(r)|_{\rho})>0$
or $\text{DT}_{\text{depth}}(\Xi_{t}(r)|_{\rho})\geq1$. 
\end{proof}
We define the probabilities that $\Xi_{t}(r)$ evaluates equivalently
to 0, or 1, or has $\text{DT}_{\text{depth}}\geq1$ respectively. 
\begin{defn}
\label{def:p0,p1,pstar}Let $\Xi_{t}(r$) be a $t$-clipped xor tribe
on $r$ levels and $\rho$ a random $p$-restriction. Let $P_{0}(r)$
be the probability that $\Xi_{t}(r$) evaluates to the 0 function
after being hit by a random $p$-restriction $\rho$, $P_{1}(r)$
the probability that $\Xi_{t}(r$) evaluates to the 1 function after
being hit by $\rho$, and $P_{*}(r)$ the probability that $\Xi_{t}(r)$
has $\text{DT}_{\text{depth}}\geq1$ after being hit by a random restriction.
More formally, letting the parameters $p$ and $t$ be implicit,
\begin{eqnarray}
P_{0}(r) & = & \Pr_{\rho}[f(\Xi_{t}(r)|_{\rho})\equiv0],\nonumber \\
P_{1}(r) & = & \Pr_{\rho}[f(\Xi_{t}(r)|_{\rho})\equiv1],\nonumber \\
P_{*}(r) & = & \Pr_{\rho}[\text{DT}_{\text{depth}}(\Xi_{t}(r)|_{\rho})\geq1]\nonumber \\
 & = & 1-P_{0}(r)-P_{1}(r).\label{eq:p_star_r}
\end{eqnarray}
Given Lemma \ref{claim:obs_path_to_0_1} and definition \ref{def:p0,p1,pstar},
our proof strategy for Lemma \ref{lem:d_1_lowerbound} is as follows:
\end{defn}
\noindent \textbf{Step 1:} We write the exact expressions for $P_{0}(r)$
and $P_{1}(r)$ as polynomials in $p$, in Lemmas \ref{lem:p01_and_p11}
and \ref{lem:p0r_and_p1r}. The expressions are obtained by counting
arguments and using recursion.\\

\noindent \textbf{Step 2:} In the next step, we reason about the constant
coefficients in $P_{0}(r)$ and $P_{1}(r)$. We derive the expressions
for $[1]P_{0}(r)$ and $[1]P_{1}(r)$ in Lemmas \ref{lem:coeff_consts_rec}
and \ref{lem:coeff_const_exact}, and show the following in Lemma
\ref{lem:sum_constcoeff_1}:
\[
[1]P_{0}(r)+[1]P(r)=1.
\]

\noindent \textbf{Step 3:} In the third step, we first compute the
recursive expressions for $[p]P_{0}(r)$ and $[p]P_{1}(r)$ in Lemma
\ref{lem:coeff_p}, and in Lemmas \ref{lem:pcoeff_upperbound} and
\ref{lem:lem_main_1}, show that for $r\in\Omega(2^{t}),$
\[
-4\cdot2^{t}\leq[p](P_{0}(r)+P_{1}(r))\leq-\frac{1}{6}2^{t}.
\]

\noindent \textbf{Step 4:} We show that higher powers of $p$ do not
substantially affect the coefficient of $p$. In Lemma \ref{lem:main_higherpowers},
we show that for all $r$, for $0\leq p\leq c_{p}2^{-t}$ where $c_{p}=\frac{1}{420}$,
\[
G_{2}(P_{0}(r)+P_{1}(r))\leq30\cdot2^{2t}.
\]

\noindent \textbf{Step 5:} Using the conclusions of Steps 2,3,4, we
infer Lemma \ref{lem:d_1_lowerbound}.\\

We start by writing the recurrence relations for $P_{0}(r)$ and $P_{1}(r)$.

\subsubsection{Recurrence relations}

Note that we set independently set a variable to $*$ with probability
$p$, and to $0$ or $1$ with probability $q=\frac{1}{2}(1-p).$ 
\begin{lem}
\label{lem:p01_and_p11} For the base case, $P_{0}(1)$ and $P_{1}(1)$
are
\begin{eqnarray}
P_{0}(1) & = & q^{t},\label{eq:base_p0}\\
P_{0}(1) & = & 1-(1-q)^{t}.\label{eq:base_p1}
\end{eqnarray}
\end{lem}
\begin{proof}
Let the tree tribe be $\Xi_{t}(1)$. Let the variables queried be
$\{x_{1},\ldots,x_{t}\}$. Note that for this case of $r=1$, $\Xi_{t}(1)$
behaves just like the OR function on the variables. Thus, after being
hit by a random restriction, the tree evaluates to the 0 function
only if all the variables are assigned 0 by $\rho$. If there is any
variable assigned $*$ or $1$ by $\rho$, the tree cannot evaluate
to $0$. Thus,
\[
P_{0}(1)=q^{t}.
\]
Computing $P_{1}(1)$ is simple and will illustrate observation \ref{claim:obs_path_to_0_1}.
In this case, it is possible that some variables are assigned $*$
by $\rho$, however, the function evaluates to 1. As an example, consider
the restriction such that $\rho(x_{1})=*$ and $\rho(x_{2})=1$. Here,
although $x_{1}$ is assigned $*$, the function evaluates to 1, specifically
because the root does not have a path to a 0-leaf and a path to a
1-leaf. To compute $P_{1}(1)$ succinctly, let $P_{1,t}(1)$ be the
probability that $\Xi_{t}(1)$ evaluates to 1, i.e.,
\[
P_{1,t}(1)=\Pr_{\rho}[f(\Xi_{t}(1)|_{\rho})\equiv1].
\]
We have made the parameter $t$ explicit to ease the computation of
$P_{1}(1)$. Note that $P_{1,t}(1)=P_{1}(1)$ and $P_{1,1}(1)=q$.
The subtree connected to the 0-edge of $x_{1}$ is exactly $\Xi_{t-1}(1)$,
and along the 1-edge of $x_{1}$ is the leaf labelled 1. Then, writing
the recursion for $P_{1,t}(1)$, we get,
\begin{eqnarray*}
P_{1,t}(1) & = & q+(p+q)P_{1,t-1}(1)\\
 & = & q\sum_{i=0}^{t-2}(p+q)^{i}+(p+q)^{t-1}P_{1,1}(1)\\
 & = & 1-(1-q)^{t}
\end{eqnarray*}
And thus, 
\[
P_{0}(1)=1-(1-q)^{t},
\]
as required.\end{proof}
\begin{lem}
\label{claim:simple_claim_for_recurrence}Let the variables in the
first level of $\Xi_{t}(r)$ be $X=\{x_{1},\ldots,x_{t}\}$, and the
$t$ children out of each of their 1-edges be $Q_{1},\ldots,Q_{t}$
respectively. Let $Y_{0}\subseteq X$ and $Y_{*}\subseteq X$ be subsets
of variables assigned $0$ and $*$ respectively by $\rho$ where
$Y_{0}\cap Y_{*}=\emptyset$, such that the index of every variable
in $Y_{0}$ and $Y_{*}$ is less than $j$, and $\rho(x_{j})=1$.
Then $\Xi_{t}(r)|_{\rho}\equiv0$ if and only if $Q_{i}|_{\rho}\equiv0$
for all $x_{i}\in Y_{*}$ and $Q_{j}|_{\rho}\equiv0$. \end{lem}
\begin{proof}
For any $i<j$ where $i$ is the least index such that $x_{i}\in Y_{*}$,
if $Q_{i}|_{\rho}\not\equiv0$ and $Q_{i}|_{\rho}\not\equiv1$, then
$Q_{i}|_{\rho}$ requires a decision tree of depth at least 1, and
by observation \ref{claim:obs_path_to_0_1}, this means that the root
of $Q_{i}|_{\rho}$ has a path to both a 0-leaf and a 1-leaf. Since
the variables of $Q_{i}$ are all different from $X$ and thus $Y_{*}$,
this implies that $x_{1}$ has a path to some 0 and 1 leaves, and
thus $\Xi_{t}(r)|_{\rho}$ requires a decision tree of depth at least
1. Consider the case where each of $Q_{i}|_{\rho}\equiv0$ or $Q_{i}|_{\rho}\equiv1$
for all $x_{i}\in Y_{*}$. But if for some $x_{i}$, $Q_{i}|_{\rho}\equiv1$,
then the function $\Xi_{t}(r)|_{\rho}$ can evaluate to 1 when $x_{i}$
is queried by some decision tree, and $\Xi_{t}(r)|_{\rho}\not\equiv0$.
In case $Y_{*}=\emptyset$, the claim trivially follows. Exactly by
the same argument, it is required that $Q_{j}|_{\rho}\equiv0$ for
$\Xi_{t}(r)|_{\rho}\equiv0$.\end{proof}
\begin{lem}
\label{lem:p0r_and_p1r} Let $U=P_{1}(r-1)$ and $V=P_{0}(r-1)$.
Then, 
\begin{eqnarray}
P_{0}(r) & = & \sum_{k=0}^{t-1}qU\left(q+pU\right)^{k}+\left(q+pU\right)^{t}\label{eq:p0r}\\
 & = & qU\frac{1-(q+pU)^{t}}{1-(q+pU)}+(q+pU)^{t}\nonumber 
\end{eqnarray}
and
\begin{eqnarray}
P_{1}(r) & = & \sum_{k=0}^{t-1}qV\left(q+pV\right)^{k}\label{eq:p1r}\\
 & = & qV\frac{1-(q+pV)^{t}}{1-(q+pV)}.\nonumber 
\end{eqnarray}
\end{lem}
\begin{proof}
Let the variables in the first level of $\Xi_{t}(r)$ be $X=\{x_{1},\ldots,x_{t}\}$,
and the $t$ children out of each of their 1-edges be $Q_{1},\ldots,Q_{t}$
respectively, where note that by construction, each of the functions
$Q_{i}=\neg\Xi_{t}(r-1)$. Thus, we have that for all $i\in\{1,\ldots,t\}$,
\begin{eqnarray}
\Pr[(Q_{i}|_{\rho})\equiv0] & = & P_{1}(r-1),\label{eq:temp_simple_eqs}\\
\Pr[(Q_{i}|_{\rho})\equiv1] & = & P_{0}(r-1).\nonumber 
\end{eqnarray}
Note that since the variables in each $Q_{i}$ are different, we have
complete independence between the probabilities that they evaluate
to 0 or 1, i.e., 
\[
\Pr_{\rho}[Q_{i_{1}}|_{\rho}\equiv a_{i_{1}},\ldots,Q_{i_{l}}|_{\rho}\equiv a_{i_{l}}]=\Pr_{\rho}[Q_{i_{1}}|_{\rho}\equiv a_{i_{1}}]\ldots\Pr_{\rho}[Q_{i_{l}}|_{\rho}\equiv a_{i_{l}}]
\]
where $a_{i_{j}}\in\{0,1\}$ for $j\in\{1,\ldots,l\}$. To compute
$P_{0}(r)$, we partition the events based on the least index $k\in\{1,\ldots,t\}$
such that $\rho(x_{k})=1$, and the case that no variable $x_{k}$
is assigned 1 by $\rho$. If $k$ is the least index such that $\rho(x_{k})=1$,
then note that for every $i\in\{0,\ldots,k-1\}$, there is a possibility
of having $i$ variables assigned $*$ by $\rho$ in $\{x_{1},\ldots,x_{k-1}\}$,
and the rest assigned 0. To compute $P_{0}(r)$, we would require
that for every $Q_{i}$ where $i\in\{1,\ldots,k-1\}$ and $\rho(x_{i})=*$,
$f(Q_{i})\equiv0$. Further, we would require that $f(Q_{k})\equiv0$.
Summing up the probabilities by partitioning events in this manner,
using \ref{eq:temp_simple_eqs} and claim \ref{claim:simple_claim_for_recurrence},
we get that,

\begin{eqnarray*}
P_{0}(r) & = & \sum_{k=0}^{t-1}qP_{1}(r-1)\sum_{i=0}^{k}{k \choose i}q^{k-i}p^{i}\left(P_{1}(r-1)\right)^{i}+\sum_{i=0}^{t}{t \choose i}q^{t-i}p^{i}\left(P_{1}(r-1)\right)^{i}\\
 & = & \sum_{k=0}^{t-1}qP_{1}(r-1)\left(q+pP_{1}(r-1)\right)^{k}+\left(q+pP_{1}(r-1)\right)^{t}\\
 & = & qU\frac{1-(q+pU)^{t}}{1-(q+pU)}+(q+pU)^{t}.
\end{eqnarray*}

To compute $P_{1}(r)$, note that in this case, at least one of the
variables in $x_{1},\ldots,x_{t}$ must be assigned $1$ by $\rho$,
since if all the variables are assigned either 0 or $*$ by $\rho$,
then by definition of $\Xi_{t}(r)$, there is always a possibility
of the remaining function to evaluate to 0 if all the variables that
are assigned $*$ by $\rho$ take the value 0. All other details remain
exactly the same, and we get,

\begin{eqnarray*}
P_{1}(r) & = & \sum_{k=0}^{t-1}qP_{0}(r-1)\sum_{i=0}^{k}{k \choose i}q^{k-i}p^{i}\left(P_{0}(r-1)\right)^{i}\\
 & = & \sum_{k=0}^{t-1}qP_{0}(r-1)\left(q+pP_{0}(r-1)\right)^{k}\\
 & = & qV\frac{1-(q+pV)^{t}}{1-(q+pV)}.
\end{eqnarray*}

\end{proof}
All of $P_{0}(r)$, $P_{1}(r)$ and $P_{*}(r)$ in equations \ref{eq:p0r},
\ref{eq:p1r} and \ref{eq:p_star_r} are univariate polynomials in
$p$. We reason about the constant coefficients in $P_{0}(r)$ and
$P_{1}(r)$, i.e. $[1]P_{0}(r)$ and $[1]P_{1}(r)$. The straightforward
proofs are given in appendix \ref{sec:app_low_bound_proofs}.

\subsubsection{The constant coefficients, $[1]P_{0}(r)$ and $[1]P_{1}(r)$}
\begin{lem}
\label{lem:coeff_consts_rec}The recursive expressions for $[1]P_{0}(r)$
and $[1]P_{1}(r)$ are as follows:
\begin{eqnarray}
[1]P_{0}(1) & = & \frac{1}{2^{t}}\label{eq:p01_const}\\{}
[1]P_{1}(1) & = & 1-\frac{1}{2^{t}},\label{eq:p11_const}
\end{eqnarray}
and
\begin{eqnarray}
[1]P_{0}(r) & = & \left(1-\frac{1}{2^{t}}\right)[1]P_{1}(r-1)+\frac{1}{2^{t}},\label{eq:const_p0}\\{}
[1]P_{1}(r) & = & \left(1-\frac{1}{2^{t}}\right)[1]P_{0}(r-1).\label{eq:const_p1}
\end{eqnarray}
\end{lem}
\begin{proof}
The proof is given as Lemma \ref{lem:app_coeff_consts_rec} in the
appendix. 
\end{proof}
Using the recurrence relations in Lemma \ref{lem:coeff_consts_rec},
we can prove the following simple but useful fact. 
\begin{lem}
\label{lem:sum_constcoeff_1}$[1]P_{0}(r)+[1]P(r)=1$ and $[1]P_{0}(r)\geq0$,
$[1]P_{1}(r)\geq0$. \end{lem}
\begin{proof}
The fact that the coefficients are always non-negative follows simply
from equations \ref{eq:const_p0} and \ref{eq:const_p1} via straightforward
induction on $r$. We use induction on $r$ again to show that the
coefficients sum to 1. For $r=1$, from equations \ref{eq:p01_const}
and \ref{eq:p11_const}, 
\[
[1]P_{0}(1)+[1]P_{1}(1)=\frac{1}{2^{t}}+1-\frac{1}{2^{t}}=1.
\]
Adding equations \ref{eq:const_p0} and \ref{eq:const_p1}, we get
\[
[1]P_{0}(r)+[1]P_{1}(r)=\left(1-\frac{1}{2^{t}}\right)\left([1]P_{1}(r-1)+[1]P_{0}(r-1)\right)+\frac{1}{2^{t}}
\]
and the induction follows. 
\end{proof}
We compute the exact equations for $[1]P_{0}(r)$ and $[1]P_{1}(r)$,
which will be useful when we compute the coefficients of higher powers
of $p$.
\begin{lem}
\label{lem:coeff_const_exact} Let $\alpha=2^{-t}$ and $\beta=1-\alpha$.
Then the exact expressions for $[1]P_{0}(r)$ and $[1]P_{1}(r)$ are
as follows. If $r=2k+1$, 
\begin{eqnarray}
[1]P_{0}(r) & = & \frac{1-\beta^{r+1}}{1+\beta},\label{eq:exact_const_p0_odd}\\{}
[1]P_{1}(r) & = & \frac{\beta+\beta^{r+1}}{1+\beta},\label{eq:exact_const_p1_odd}
\end{eqnarray}
and if $r=2k$, 
\begin{eqnarray}
[1]P_{0}(r) & = & \frac{1+\beta^{r+1}}{1+\beta},\label{eq:exact_const_p0_even}\\{}
[1]P_{1}(r) & = & \frac{\beta-\beta^{r+1}}{1+\beta}.\label{eq:eq:exact_const_p1_even}
\end{eqnarray}
\end{lem}
\begin{proof}
The proof is given as Lemma \ref{lem:app_coeff_consts} in the appendix. 
\end{proof}
We bound the coefficient of $p$ in $P_{0}(r)$ and $P_{1}(r)$.

\subsubsection{The coefficients of $p$, $[p]P_{0}(r)$ and $[p]P_{1}(r)$}
\begin{lem}
\label{lem:coeff_p}The recurrence relations for $[p]P_{0}(r)$ and
$[p]P_{1}(r)$ are as follows:
\begin{eqnarray}
[p]P_{0}(1) & = & -\frac{t}{2^{t}}\label{eq:p_p01}\\{}
[p]P_{1}(1) & = & -\frac{t}{2^{t}}.\label{eq:p_p11}
\end{eqnarray}
and
\begin{eqnarray}
[p]P_{1}(r) & = & \left(1-\frac{1}{2^{t}}\right)[p]P_{0}(r-1)+\left(\frac{t+2}{2^{t}}-2\right)[1]P_{0}(r-1)\nonumber \\
 &  & +2\left(1-\frac{t+1}{2^{t}}\right)[1]P_{0}(r-1){}^{2}\label{eq:pcoeffp1}\\{}
[p]P_{0}(r) & = & \left(1-\frac{1}{2^{t}}\right)[p]P_{1}(r-1)+\left(\frac{t+2}{2^{t}}-2\right)[1]P_{1}(r-1)\nonumber \\
 &  & +2\left(1-\frac{t+1}{2^{t}}\right)[1]P_{1}(r-1){}^{2}-\frac{t}{2^{t}}+\frac{2t}{2^{t}}[1]P_{1}(r-1).\label{eq:pcoeffp0}
\end{eqnarray}
\end{lem}
\begin{proof}
The proof is given as Lemma \ref{lem:app_coeff_p} in the appendix. 
\end{proof}
Given the recurrence relations for $[p]P_{0}(r)$ and $[p]P_{1}(r)$,
although it is possible to compute the exact expressions for them,
we will only need specific bounds on them, which we find next. 
\begin{lem}
\label{lem:pcoeff_upperbound}For every $r\geq1$, $-2\cdot2^{t}\leq[p]P_{0}(r)\leq0$
and $-2\cdot2^{t}\leq[p]P_{1}(r)\leq0$.\end{lem}
\begin{proof}
The proof is given as Lemma \ref{lem:app_p_lowbound} in the appendix. 
\end{proof}
The next lemma is important, since it gives the coefficient of $p$
that we need for the lower bound of Lemma \ref{lem:d_1_lowerbound}. 
\begin{lem}
\label{lem:lem_main_1}For $r\in\Omega(2^{t})$, $[p](P_{0}(r)+P_{1}(r))\leq-c_{1}2^{t}$
where $c_{1}=\frac{1}{6}$.\end{lem}
\begin{proof}
Define 
\[
g(r)=[p]P_{0}(r)+[p]P_{1}(r)
\]
and
\[
u(r)=[1]P_{1}(r).
\]
Note that we have $\alpha=2^{-t}$ and $\beta=1-\alpha$. Adding equations
\ref{eq:pcoeffp0} and \ref{eq:pcoeffp1}, and using Lemma \ref{lem:sum_constcoeff_1},
we get,

\begin{eqnarray*}
g(r) & = & \left(1-\frac{1}{2^{t}}\right)g(r-1)+\left(\frac{t+2}{2^{t}}-2\right)-\frac{t}{2^{t}}+\frac{2t}{2^{t}}[1]P_{1}(r-1)\\
 &  & +2\left(1-\frac{t+1}{2^{t}}\right)\left([1]P_{1}(r-1){}^{2}+[1]P_{0}(r-1)^{2}\right)\\
 & = & \beta g(r-1)+\left(\frac{2}{2^{t}}-2\right)+\frac{2t}{2^{t}}[1]P_{1}(r-1)\\
 &  & +2\left(1-\frac{t+1}{2^{t}}\right)\left(1-2[1]P_{1}(r-1)+2[1]P_{1}(r-1){}^{2}\right)\\
 & = & \beta g(r-1)-2\alpha t+(6\alpha t+4\alpha-4)u_{r-1}+(4-4\alpha-4\alpha t)u_{r-1}^{2}\\
 & = & \beta g(r-1)+C_{0}+C_{1}u_{r-1}+C_{2}u_{r-1}^{2}
\end{eqnarray*}
where $C_{0}=-2\alpha t$, $C_{1}=6\alpha t+4\alpha-4$, and $C_{2}=4-4\alpha-4\alpha t$.
Thus, 
\[
g(r)=\beta^{r-1}g(1)+C_{0}\sum_{i=0}^{r-2}\beta^{i}+C_{1}\sum_{i=0}^{r-2}\beta^{i}u_{r-i-1}+C_{2}\sum_{i=0}^{r-2}\beta^{i}u_{r-i-1}^{2}.
\]
Fix $r=2k$ to be an even number where $k\geq2$. The third term can
be written using equations \ref{eq:exact_const_p1_odd} and \ref{eq:eq:exact_const_p1_even}
as, 
\begin{eqnarray*}
\sum_{i=0}^{r-2}\beta^{i}u_{r-i-1} & = & \sum_{i=0}^{\frac{1}{2}(r-2)-1}\beta^{2i}\left(u_{r-2i-1}+\beta u_{r-2i-2}\right)+\beta^{r-2}u_{1}\\
 & = & \sum_{i=0}^{\frac{1}{2}(r-2)-1}\beta^{2i}\left(\frac{\beta+\beta^{r-2i}}{1+\beta}+\beta\frac{\beta-\beta^{r-2i-1}}{1+\beta}\right)+\beta^{r-2}u_{1}\\
 & = & \beta\frac{1-\beta^{r-2}}{1-\beta^{2}}+\beta^{r-1}\\
 & = & \beta\frac{1-\beta^{r}}{1-\beta^{2}}
\end{eqnarray*}
where in the second line we used \ref{eq:p11_const} for $u_{1}$.
Similarly, for the third summation,
\begin{eqnarray*}
\sum_{i=0}^{r-2}\beta^{i}u_{r-i-1}^{2} & = & \sum_{i=0}^{\frac{1}{2}(r-2)-1}\beta^{2i}\left(u_{r-2i-1}^{2}+\beta u_{r-2i-2}^{2}\right)+\beta^{r-2}u_{1}^{2}\\
 & = & \sum_{i=0}^{\frac{1}{2}(r-2)-1}\beta^{2i}\left(\left(\frac{\beta+\beta^{r-2i}}{1+\beta}\right)^{2}+\beta\left(\frac{\beta-\beta^{r-2i-1}}{1+\beta}\right)^{2}\right)+\beta^{r-2}u_{1}^{2}\\
 & = & \frac{\beta^{2}}{1+\beta}\sum_{i=0}^{\frac{1}{2}(r-2)-1}\beta^{2i}+\frac{\beta^{2r-1}}{1+\beta}\sum_{i=0}^{\frac{1}{2}(r-2)-1}\beta^{-2i}+\beta^{r-2}u_{1}^{2}\\
 & = & \frac{\beta^{2}}{1+\beta}\frac{(1+\beta^{r-1})(1-\beta^{r})}{1-\beta^{2}}.
\end{eqnarray*}
Thus the expression for $g(r)$ becomes,
\begin{eqnarray}
g(r) & = & \beta^{r-1}g(1)+C_{0}\frac{1-\beta^{r-1}}{1-\beta}+C_{1}\beta\frac{1-\beta^{r}}{1-\beta^{2}}+C_{2}\frac{\beta^{2}}{1+\beta}\frac{(1+\beta^{r-1})(1-\beta^{r})}{1-\beta^{2}}\nonumber \\
 & = & \beta^{r-1}g(1)+(-2\alpha t)\frac{1-\beta^{r-1}}{1-\beta}+(6\alpha t+4\alpha-4)\beta\frac{1-\beta^{r}}{1-\beta^{2}}\nonumber \\
 &  & +(4-4\alpha-4\alpha t)\frac{\beta^{2}}{1+\beta}\frac{(1+\beta^{r-1})(1-\beta^{r})}{1-\beta^{2}}\nonumber \\
 & = & \beta^{r-1}g(1)+A(r)+B(r),\label{eq:gr_split}
\end{eqnarray}
where 
\begin{eqnarray*}
A(r) & = & (4\alpha-4)\beta\frac{1-\beta^{r}}{1-\beta^{2}}+(4-4\alpha)\frac{\beta^{2}}{1+\beta}\frac{(1+\beta^{r-1})(1-\beta^{r})}{1-\beta^{2}}\\
 & = & (4-4/\alpha)\beta\frac{1-\beta^{r}}{1+\beta}\left(1-\frac{\beta}{1+\beta}(1+\beta^{r-1})\right)\\
 & = & (4-4/\alpha)\beta\frac{(1-\beta^{r})^{2}}{(1+\beta)^{2}}\\
 & = & -\frac{4}{\alpha}\beta^{2}\frac{(1-\beta^{r})^{2}}{(1+\beta)^{2}},
\end{eqnarray*}
and
\begin{eqnarray*}
B(r) & = & (-2\alpha t)\frac{1-\beta^{r-1}}{1-\beta}+(2\alpha t)\beta\frac{1-\beta^{r}}{1-\beta^{2}}\\
 &  & +(4\alpha t)\beta\frac{1-\beta^{r}}{1-\beta^{2}}+(-4\alpha t)\frac{\beta^{2}}{1+\beta}\frac{(1+\beta^{r-1})(1-\beta^{r})}{1-\beta^{2}}\\
 & = & (-2t)\left(\frac{1-\beta^{r-1}-\beta^{r}+\beta^{r+1}}{1+\beta}\right)\\
 &  & +(4t)\beta\frac{1-\beta^{r}}{1+\beta}\left(1-\frac{\beta}{1+\beta}(1+\beta^{r-1})\right)\\
 & = & (-2t)\left(\frac{1-\beta^{r-1}-\beta^{r}+\beta^{r+1}}{1+\beta}\right)+4t\beta\frac{(1-\beta^{r})^{2}}{(1+\beta)^{2}}\\
 & = & \frac{2t}{(1+\beta)^{2}}\left(-(1+\beta)(1-\beta^{r-1}-\beta^{r}+\beta^{r+1})+2\beta(1-\beta^{r})^{2}\right)\\
 & = & \frac{2t}{(1+\beta)^{2}}(\beta-1+2\beta^{r}-4\beta\beta^{r}+\beta^{r-1}+2\beta^{2r+1}-\beta^{r+2})\\
 & \leq & \frac{2t\beta^{r}}{(1+\beta)^{2}}(\frac{1}{\beta}+2\beta^{r+1}-\beta^{2})\\
 & \leq & \frac{8t\beta^{r}}{(1+\beta)^{2}},
\end{eqnarray*}
where in the first inequality, we used the facts that since $\frac{1}{2}\leq\beta=1-2^{-t}\leq1$,
we would have $\beta-1\leq0$ and $1-2\beta\leq0$, and in the second
inequality we used the same bounds for $\beta$ after factoring out
$\beta^{r}$. Replacing the values of $A(r)$ and $B(r)$ in equation
\ref{eq:gr_split} for $g(r)$, and using that $\beta\geq0$ and $g(1)\leq0$
from Lemma \ref{lem:coeff_p}, we get that 
\begin{eqnarray*}
g(r) & = & \beta^{r-1}g(1)+A(r)+B(r)\\
 & \leq & -\frac{4}{\alpha}\beta^{2}\frac{(1-\beta^{r})^{2}}{(1+\beta)^{2}}+\frac{8t\beta^{r}}{(1+\beta)^{2}}\\
 & = & -\frac{3}{\alpha}\beta^{2}\frac{(1-\beta^{r})^{2}}{(1+\beta)^{2}}-\frac{1}{\alpha}\beta^{2}\frac{(1-\beta^{r})^{2}}{(1+\beta)^{2}}+\frac{8t\beta^{r}}{(1+\beta)^{2}}
\end{eqnarray*}
Taking $r=C2^{t}$ for some constant $C$, we get that 
\[
\beta^{r}=\left(1-\frac{1}{2^{t}}\right)^{C2^{t}}\leq e^{-C},
\]
and using $\alpha=2^{-t}$ and $\frac{1}{2}\leq\beta\leq1$, we get
\begin{eqnarray*}
-\frac{1}{\alpha}\beta^{2}\frac{(1-\beta^{r})^{2}}{(1+\beta)^{2}}+\frac{8t\beta^{r}}{(1+\beta)^{2}} & \leq & -2^{t}\frac{1}{4}\frac{(1-\beta^{r})^{2}}{(1+\beta)^{2}}+\frac{8t\beta^{r}}{(1+\beta)^{2}}\\
 & \leq & -\frac{t}{4(1+\beta)^{2}}\left((1-\beta^{r})^{2}-32\beta^{r}\right)\\
 & \leq & -\frac{t}{4(1+\beta)^{2}}\left((1-e^{-C})^{2}-32e^{-C}\right)\\
 & \leq & 0
\end{eqnarray*}
where in the second line, we used $t\leq2^{t}$, and in the last line,
we used the fact that for $C=4$, $(1-e^{-C})^{2}-32e^{-C}\geq0$.
Thus we get that 
\begin{eqnarray*}
g(r) & \leq & -\frac{3}{\alpha}\beta^{2}\frac{(1-\beta^{r})^{2}}{(1+\beta)^{2}}\\
 & \leq & -3\cdot2^{t}\frac{1}{4}\frac{1}{4}(1-e^{-C})^{2}\\
 & \leq & -\frac{1}{6}2^{t}
\end{eqnarray*}
where in the second line, we used that for $C=4$, $(1-e^{-C})^{2}\geq\frac{16}{18}$. 
\end{proof}
Thus so far, we have the following facts about the polynomial $P_{0}(r)+P_{1}(r)$,
as was mentioned in Steps 2 and 3. 
\begin{enumerate}
\item The constant coefficient is 1, i.e. $[1]P_{0}(r)+[1]P_{1}(r)=1$,
and
\item The coefficient of $p$ is between $-4\cdot2^{t}$ and $-c_{1}2^{t}$
for $r\in\Omega(2^{t})$ and $c_{1}=\frac{1}{6}$, i.e., 
\[
-4\cdot2^{t}\leq[p]P_{0}(r)+[p]P_{1}(r)\leq-c_{1}2^{t}.
\]

\end{enumerate}
Note that we almost have sufficient information to conclude Lemma
\ref{lem:d_1_lowerbound}, however, we need to reason that higher
powers of $p$ in $P_{0}(r)+P_{1}(r)$ cannot substantially affect
the coefficient of $p$. It is possible that higher powers of $p$
have coefficients with large magnitudes, and thus annihilate the mass
on the coefficient of $p$. We show next that this cannot happen.
Here we will use the fact that $p\leq c_{p}2^{-t}$.

\subsubsection{Bounding higher powers of $p$ in $P_{0}(r)+P_{1}(r)$}

To bound the higher powers of $p$, we will show that $[\uparrow p^{2}]\left(P_{0}(r)+P_{1}(r)\right)$
cannot achieve a very large value for $p\leq c_{p}2^{-t}$. Thus,
our aim is to show that $G_{2}(P_{0}(r)+P_{1}(r))$ is small. 
\begin{lem}
\label{lem:main_higherpowers}For every $r\geq1$, for $0\leq p\leq p_{\max}=c_{p}2^{-t}$,
where $c_{p}=\frac{1}{420}$, 
\[
G_{2}(P_{0}(r)+P_{1}(r))\leq30\cdot2^{2t}.
\]
\end{lem}
\begin{proof}
We will assume that 
\[
p\leq p_{\max}=\frac{1}{C2^{t}}
\]
where $C=\frac{1}{c_{p}}$. Using part (6) of claim \ref{claim:poly_props},
\[
G_{2}(P_{0}(r)+P_{1}(r))\leq G_{2}(P_{0}(r))+G_{2}(P_{1}(r))
\]
and we will bound the terms individually. We will show this fact by
induction. For the base case $r=1$, using claim \ref{claim:basic_ineqs},
\[
G_{2}(P_{0}(1))\leq\max_{p}\frac{1}{2^{t}}\frac{1}{p^{2}}\left((1-p)^{t}-1+tp\right)\leq\frac{1}{2^{t}}{t \choose 2}\leq30\cdot2^{2t},
\]
and for $P_{1}(1)$,
\begin{eqnarray*}
G_{2}(P_{1}(1)) & \leq & \max_{p}\left|[\uparrow p^{2}]\left(1-\frac{1}{2^{t}}(1+p)^{t}\right)\right|\\
 & = & \max_{p}\left|-\frac{1}{2^{t}}\sum_{i=2}^{t}{t \choose i}p^{i-2}\right|\\
 & \leq & \frac{1}{2^{t}}\sum_{i=2}^{t}{t \choose i}\\
 & \leq & 1,
\end{eqnarray*}
where we used the fact that $p\leq1$ in the third line. Assume that
for all $1\leq z\leq r$, 
\begin{eqnarray*}
G_{2}(P_{0}(z)) & \leq & \Gamma,\\
G_{2}(P_{1}(z)) & \leq & \Gamma.
\end{eqnarray*}
Reproducing equations \ref{eq:p0r} and \ref{eq:p1r} for $P_{0}(r)$
and $P_{1}(r)$, for $U=P_{1}(r-1)$ and $V=P_{0}(r-1)$, 
\begin{eqnarray*}
P_{0}(r) & = & qU\sum_{i=0}^{t-1}(q+pU)^{i}+(q+pU)^{t},\\
P_{1}(r) & = & qV\sum_{i=0}^{t-1}(q+pV)^{i}.
\end{eqnarray*}
Let us analyze one term of $P_{0}(r)$, $qU(q+pU)^{k}$. Note that
\begin{eqnarray*}
[\uparrow p^{2}]qU(q+pU)^{k} & = & [\uparrow p^{2}]\sum_{i=0}^{k}{k \choose i}q^{k-i+1}p^{i}U^{i+1}\\
 & = & \frac{1}{p^{2}}\left(qU(q+pU)^{k}-q^{k+1}U-kq^{k}pU^{2}\right)+[\uparrow p^{2}]q^{k+1}U+[\uparrow p]kq^{k}U^{2}.\\{}
[\uparrow p^{2}]q^{k+1}U & = & \frac{1}{p^{2}}\left(\frac{U}{2^{k+1}}\left((1-p)^{k+1}-1+(k+1)p\right)\right)+[\uparrow p^{2}]\frac{U}{2^{k+1}}-[\uparrow p]\frac{k+1}{2^{k+1}}U.\\{}
[\uparrow p]kq^{k}U^{2} & = & \frac{1}{p}\left(\frac{kU^{2}}{2^{k}}\left((1-p)^{k}-1\right)\right)+[\uparrow p]\frac{k}{2^{k}}U^{2}.
\end{eqnarray*}
To analyse each of these terms, let the functions be as follows:
\begin{eqnarray*}
h_{1}^{k} & = & \frac{1}{p^{2}}\left(qU(q+pU)^{k}-q^{k+1}U-kq^{k}pU^{2}\right)\\
h_{2}^{k} & = & \frac{1}{p^{2}}\left(\frac{U}{2^{k+1}}\left((1-p)^{k+1}-1+(k+1)p\right)\right)\\
h_{3}^{k} & = & \frac{1}{p}\left(\frac{kU^{2}}{2^{k}}\left((1-p)^{k}-1\right)\right)\\
h_{4}^{k} & = & [\uparrow p^{2}]\frac{U}{2^{k+1}}\\
h_{5}^{k} & = & [\uparrow p]\frac{k}{2^{k}}U^{2}\\
h_{6}^{k} & = & -[\uparrow p]\frac{k+1}{2^{k+1}}U
\end{eqnarray*}
Consider the function $h_{1}^{k}$. 
\begin{eqnarray*}
h_{1}^{k} & = & \frac{1}{p^{2}}\left(qU(q+pU)^{k}-q^{k+1}U-kq^{k}pU^{2}\right)\\
 & = & \sum_{i=2}^{k}{k \choose i}q^{k-i+1}p^{i-2}U^{i+1}.
\end{eqnarray*}
Since $h_{1}^{k}\geq0$, $|h_{1}^{k}|=h_{1}^{k}$. Further, since
$q\geq0$ and $p\geq0$, and the function above is an increasing function
of $U=P_{1}(r-1)\leq1$, we can set $U=1$. This step is innocuous
but crucial, since there are large cancellations that happen within
the powers of $p$ in the expression for $U$. Thus, we can bound,
\begin{eqnarray*}
\max_{p}|h_{1}^{k}| & \leq & \sum_{i=2}^{k}{k \choose i}q^{k-i+1}p^{i-2}\\
 & \leq & {k \choose 2}\frac{1}{2^{k-1}}+p\sum_{i=3}^{k}{k \choose i}q^{k-i+1}p^{i-3}\\
 & \leq & {k \choose 2}\frac{1}{2^{k-1}}+p\sum_{i=3}^{k}{k \choose i}\\
 & \leq & {k \choose 2}\frac{1}{2^{k-1}}+\frac{1}{C}\frac{2^{k}}{2^{t}}
\end{eqnarray*}
where in the second line, we used $q\leq\frac{1}{2}$, in the third
line, we used $q\leq1$ and $p\leq1$, and in the last line we used
$p\leq\frac{1}{C2^{t}}$. By a similar argument, noting that $(1-p)^{k+1}-1+(k+1)p\geq0$
for $k\geq0$ and $p\in[0,1]$, we get that $|h_{2}^{k}|=h_{2}^{k}$,
and setting $U=1$ and using point (6) in claim \ref{claim:basic_ineqs},
we get that,

\[
\max_{p}|h_{2}^{k}|\leq{k+1 \choose 2}\frac{1}{2^{k+1}}.
\]
For $h_{3}^{k}$, note that $(1-p)^{k}-1\leq0$ for $k\geq0$ and
$p\in[0,1]$, and thus $|h_{3}^{k}|=-h_{3}^{k}$. Again setting $U=1$
and using point (5) in claim \ref{claim:basic_ineqs}, 
\[
\max_{p}|h_{3}^{k}|\leq\frac{k^{2}}{2^{k}}.
\]
Using the induction hypothesis for $h_{4}^{k}$, we get that 
\[
\max_{p}|h_{4}^{k}|\leq\frac{\Gamma}{2^{k+1}}.
\]
To bound the maximum of $h_{5}^{k}$ and $h_{6}^{k}$, we can write,
\begin{eqnarray*}
U & = & P_{1}(r-1)\\
 & = & \alpha_{0}+\alpha_{1}p+p^{2}Q
\end{eqnarray*}
where $\alpha_{0}$ and $\alpha_{1}$ are constants and $Q$ is a
polynomial in $p$, such that 
\begin{eqnarray*}
0 & \leq & |\alpha_{0}|\leq1\\
|\alpha_{1}| & \leq & 2\cdot2^{t}\\
\max_{p}|Q| & \leq & \Gamma
\end{eqnarray*}
which follow respectively from Lemmas \ref{lem:sum_constcoeff_1},
\ref{lem:pcoeff_upperbound} and the induction hypothesis. Thus, we
can write, 
\begin{eqnarray}
[\uparrow p]U & = & \alpha_{1}+pQ\nonumber \\
G_{1}(U)=\max_{p}\left|[\uparrow p]U\right| & \leq & |\alpha_{1}|+p_{\max}\Gamma=B_{1},\label{eq:m_p_up}
\end{eqnarray}
and 
\begin{eqnarray}
[\uparrow p]U^{2} & = & (2\alpha_{0}\alpha_{1})+p(\alpha_{1}^{2}+\alpha_{0}Q)+p^{2}(2\alpha_{1}Q)+p^{3}(Q^{2})\nonumber \\
\max_{p}\left|[\uparrow p]U^{2}\right| & \leq & 2|\alpha_{1}|+p_{\max}(|\alpha_{1}|^{2}+\Gamma)+p_{\max}^{2}(2|\alpha_{1}|\Gamma)+p_{\max}^{3}\Gamma^{2}=B_{2}.\label{eq:m2_p_up}
\end{eqnarray}
Thus,
\begin{eqnarray*}
\max_{p}|h_{5}^{k}| & \leq & B_{2}\frac{k}{2^{k}},\\
\max_{p}|h_{6}^{k}| & \leq & B_{1}\frac{k+1}{2^{k+1}}.
\end{eqnarray*}
Summing up all the $h^{k}$ for different values of $k$ gives the
maximum of $\left|[\uparrow p^{2}]P_{0}(r)\right|$, except for the
last term $(q+pU)^{t}$, which we analyse now, in exactly the same
manner as the functions analyzed above. \footnote{The fact that the factor $qU$ does not appear in the last term, i.e.,
the last term is not $qU(q+pU)^{t}$ but only $(q+pU)^{t}$ is the
main reason our calculations work out. And note that this factor does
not appear in the last term exactly because the tree is $t$-clipped.}

\[
[\uparrow p^{2}](q+pU)^{t}=h_{1}^{t}+h_{2}^{t}+h_{3}^{t}+h_{4}^{t}+h_{5}^{t}+h_{6}^{t}
\]
where 
\begin{eqnarray*}
h_{1}^{t} & = & \frac{1}{p^{2}}\left((q+pU)^{t}-q^{t}-tq^{t-1}pU\right)\\
h_{2}^{t} & = & \frac{1}{p^{2}}\left(\frac{1}{2^{t}}\left((1-p)^{t}-1+tp\right)\right)\\
h_{3}^{t} & = & \frac{1}{p}\left(\frac{tU}{2^{t-1}}\left((1-p)^{t}-1\right)\right)\\
h_{4}^{t} & = & [\uparrow p^{2}]\frac{1}{2^{t}}\\
h_{5}^{t} & = & -[\uparrow p]\frac{t}{2^{t}}\\
h_{6}^{t} & = & [\uparrow p]\frac{t}{2^{t-1}}U.
\end{eqnarray*}
Taking bounds on them in exactly the same manner as shown before,
we get that,
\begin{eqnarray*}
\max_{p}|h_{1}^{t}| & \leq & {t \choose 2}\frac{1}{2^{t-2}}+\frac{1}{C}\\
\max_{p}|h_{2}^{t}| & \leq & {t \choose 2}\frac{1}{2^{t}}\\
\max_{p}|h_{3}^{t}| & \leq & \frac{t^{2}}{2^{t-1}}\\
\max_{p}|h_{4}^{t}| & = & 0\\
\max_{p}|h_{5}^{t}| & = & 0\\
\max_{p}|h_{6}^{t}| & \leq & B_{1}\frac{t}{2^{t-1}}.
\end{eqnarray*}
Note that, 
\begin{eqnarray*}
G_{2}(P_{0}(r)) & = & \max_{p}\left|[\uparrow p^{2}]P_{0}(r)\right|\leq\sum_{j=1}^{6}\sum_{k=0}^{t}\max_{p}|h_{j}^{k}|.
\end{eqnarray*}
Summing the first three terms over all $k$ and bounding them using
the Lemma \ref{claim:basic_ineqs}, we get,

\begin{eqnarray*}
\sum_{k=0}^{t}\max_{p}|h_{1}^{k}| & \leq & 2\sum_{k=0}^{t-1}{k \choose 2}\frac{1}{2^{k}}+\sum_{k=0}^{t-1}\frac{1}{C}\frac{2^{k}}{2^{t}}+{t \choose 2}\frac{1}{2^{t-2}}+\frac{1}{C}\leq4+\frac{2}{C}+\frac{2t^{2}}{2^{t}}\\
\sum_{k=0}^{t}\max_{p}|h_{2}^{k}| & \leq & \sum_{k=0}^{t-1}{k+1 \choose 2}\frac{1}{2^{k+1}}+{t \choose 2}\frac{1}{2^{t}}\leq2+\frac{\frac{1}{2}t^{2}}{2^{t}}\\
\sum_{k=0}^{t}\max_{p}|h_{3}^{k}| & \leq & \sum_{k=0}^{t-1}\frac{k^{2}}{2^{k}}+\frac{t^{2}}{2^{t-1}}=2\sum_{k=0}^{t-1}{k \choose 2}\frac{1}{2^{k}}+\sum_{k=0}^{t-1}\frac{k}{2^{k}}+\frac{t^{2}}{2^{t-1}}\leq6+\frac{2t^{2}}{2^{t}}
\end{eqnarray*}
and thus,
\begin{eqnarray*}
D_{0} & = & \sum_{k=0}^{t}\max_{p}|h_{1}^{k}|+\sum_{k=0}^{t}\max_{p}|h_{2}^{k}|+\sum_{k=0}^{t}\max_{p}|h_{3}^{k}|\\
 & \leq & 12+\frac{2}{C}+\frac{9}{2}\frac{t^{2}}{2^{t}}\\
 & \leq & 20,
\end{eqnarray*}
where in the last line, we assumed that $C\geq1$. For the fourth
term, we get,
\[
D_{1}=\sum_{k=0}^{t-1}\max_{p}|h_{4}^{k}|\leq\Gamma\sum_{k=0}^{t-1}\frac{1}{2^{k+1}}=\Gamma\left(1-\frac{1}{2^{t}}\right).
\]
For the last two terms, again bounding using Lemma \ref{claim:basic_ineqs},
we get, 
\[
D_{2}=\max_{p}\sum_{k=0}^{t-1}h_{5}^{k}\leq2B_{2},
\]
and 
\[
D_{3}=\max_{p}\sum_{k=0}^{t}h_{6}^{k}\leq B_{1}\left(2+\frac{t}{2^{t-1}}\right)\leq3B_{1}.
\]
We fix the constants now. Let 
\[
\Gamma=H2^{2t}
\]
where $H$ is some constant, and recall that we had set 
\[
p_{\max}=\frac{1}{C2^{t}}.
\]
With these constants, the values of $B_{1}$ and $B_{2}$ can be bound
as follows, using Lemma \ref{lem:pcoeff_upperbound} for bounding
$\alpha_{1}$: 
\begin{eqnarray*}
B_{1} & = & |\alpha_{1}|+p_{\max}\Gamma\\
 & \le & 2\cdot2^{t}+\frac{H}{C}2^{t}
\end{eqnarray*}
and
\begin{eqnarray*}
B_{2} & = & 2|\alpha_{1}|+p_{\max}(|\alpha_{1}|^{2}+\Gamma)+p_{\max}^{2}(2|\alpha_{1}|\Gamma)+p_{\max}^{3}\Gamma^{2}\\
 & \leq & 4\cdot2^{t}+\frac{1}{C}2^{t}+\frac{H}{C}2^{t}+\frac{2H}{C^{2}}2^{t}+\frac{H^{2}}{C^{3}}2^{t}.
\end{eqnarray*}
To show the induction, we need to show that
\begin{eqnarray*}
D_{0}+D_{1}+D_{2}+D_{3} & \leq & \Gamma
\end{eqnarray*}
or
\begin{eqnarray*}
20+\Gamma\left(1-\frac{1}{2^{t}}\right)+2B_{2}+3B_{1} & \leq & \Gamma
\end{eqnarray*}
which is equivalent to
\[
20+14\cdot2^{t}+\frac{5H}{C}2^{t}+\frac{2}{C}2^{t}+\frac{4H}{C^{2}}2^{t}+\frac{2H^{2}}{C^{3}}2^{t}\leq H2^{t}
\]
which after taking $2^{-t}20\leq10$ is implied by
\[
24+\frac{5H}{C}+\frac{2}{C}+\frac{4H}{C^{2}}+\frac{2H^{2}}{C^{3}}\leq H.
\]
The equation stated above is satisfied by $H=30$ and $C=420$, and
it proves the claim. Note that this gives 
\[
c_{p}=\frac{1}{C}=\frac{1}{420}
\]
in Lemma \ref{thm:lower_bound}. The proof is exactly the same for
$P_{1}(r)$ and we omit it. 
\end{proof}

\subsubsection{Lower bound}

We finally have enough tools to prove the lower bound for the $d=1$
case. We restate Lemma \ref{lem:d_1_lowerbound} here for convenience. 
\begin{lem}
\ref{lem:d_1_lowerbound}For some universal constants $c_{0}$ and
$c_{p}$, for $r\in\Omega(2^{t})$ and $0\leq p\leq c_{p}2^{-t}$,
if $\rho$ is a random $p$-restriction, then
\[
\Pr_{\rho}[\text{DT}_{\text{depth}}(\Xi_{t}(r)|_{\rho})\geq1]\geq c_{0}p2^{t}.
\]
\end{lem}
\begin{proof}
Note that the probability that the decision tree with $r$ levels
has depth more than or equal to 1, i.e. it does not become constantly
0 or 1 after being hit by a random restriction is given by 
\begin{eqnarray*}
P_{*}(r) & = & 1-P_{0}(r)-P_{1}(r).
\end{eqnarray*}
For $r\in\Omega(2^{t})$, we can write 
\begin{eqnarray*}
P_{*}(r) & = & 1-[1]\left(P_{0}(r)+P_{1}(r)\right)-\left([p]P_{0}(r)+[p]P_{1}(r)\right)p-\left([\uparrow p^{2}]\left(P_{0}(r)+P_{1}(r)\right)\right)p^{2}\\
 & = & -\left([p]P_{0}(r)+[p]P_{1}(r)\right)p-\left([\uparrow p^{2}]\left(P_{0}(r)+P_{1}(r)\right)\right)p^{2}\\
 & \geq & c_{1}2^{t}p-p_{\max}2\Gamma p\\
 & = & p2^{t}\left(\frac{1}{6}-\frac{2H}{C}\right)\\
 & \geq & \frac{1}{42}p2^{t}
\end{eqnarray*}
where in the second line we used Lemma \ref{lem:sum_constcoeff_1},
and in the third line we used Lemmas \ref{lem:lem_main_1} and \ref{lem:main_higherpowers}.
Note that we get $c_{0}=\frac{1}{42}$. 
\end{proof}

\subsection{Inductive step}

\subsubsection{The issue with taking $d=0$ as the base case}

There is an issue with using $d=0$ as the base case for induction.
Consider a decision tree $T$ with variable $x_{1}$ as the root having
subtrees $T_{0}$ and $T_{1}$ out of the 0 and 1 edges respectively.
We want to lower bound the probability of the event $E_{T,1}$, the
event that $T$ has depth greater than 1 after a random restriction
$\rho$ is applied to the variables in $T$, i.e.

\[
\Pr_{\rho}[\text{DT}_{\text{depth}}(T|_{\rho})\geq1].
\]
In Lemma \ref{claim:Has_obs}, we had that 
\[
E_{T,1}\subseteq E_{T_{0},0}\cup E_{T_{1},0}
\]
and if $x_{1}$ was assigned $*$ by $\rho$, we could upper bound
$\Pr_{\rho}[E_{T,1}]$ by the probabilities $\Pr_{\rho}[E_{T_{0},0}]$
and $\Pr_{\rho}[E_{T_{1},0}]$, both of which are 1, and there would
be no error on the upper bound on $\Pr[E_{T,1}]$. However, we cannot
lower bound $\Pr[E_{T,1}]$ as 
\[
\Pr[E_{T,1}]\geq\Pr[E_{T_{0},0}]+\Pr[E_{T_{1},0}]-\Pr[E_{T_{0},0}]\Pr[E_{T_{1},0}]
\]
This is because, after applying the random restriction to $T$, if
it assigns $*$ to $x_{1}$, even if $T_{0}$ and $T_{1}$ had depths
greater than 0 under the effect of the restriction, \emph{only if
}$T_{0}|_{\rho}\equiv1$ and $T_{1}|_{\rho}\equiv0$ or $T_{0}|_{\rho}\equiv0$
and $T_{1}|_{\rho}\equiv1$ would $T$ have depth $\geq1$. However,
the event $E_{T_{0},0}$ does not say whether $T_{0}$ evaluates to
0 or 1, and thus we cannot use $d=0$ as the base case.

\subsubsection{Lower bound for $d\geq2$}

Doing the inductive step recursively is tricky. This is because we
already need about $\sim2^{t}$ levels in the tree for the base case
to work. Thus, any induction must take care of the fact that once
the number of levels are too few, the base case would not hold. 

To write the recurrence, let $\gamma_{d}(r)$ be the probability that
$\Xi_{t}(r)$ has depth greater than or equal to $d$. 
\begin{lem}
\label{lem:gamm_rec_induc}Let $\mu=1-\gamma_{d-1}(r-1)$. Then,
\[
\gamma_{d}(r)\geq\sum_{k=1}^{t-1}q\sum_{i=1}^{k}{k \choose i}q^{k-i}p^{i}(1-\mu^{i+1})+\sum_{i=0}^{t}{t \choose i}q^{t-i}p^{i}(1-\mu^{i})+\sum_{k=0}^{t-1}q^{k+1}\gamma_{d}(r-1).
\]
\end{lem}
\begin{proof}
Consider the left most branch that takes the root to a leaf, containing
variables $X=\{x_{1},\ldots,x_{t}\}$ where $x_{1}$ is the root,
and denote $X_{i,j}=\{x_{i},\ldots,x_{j}\}$ for $i\leq j$. We will
count again by partitioning based on when the first 1 appears. The
counting happens in three steps.

1. Let $k$ be the least index such that $\rho(x_{k})=1$, and let
$S\subseteq X_{1,k-1}$ be the variables assigned $*$ by $\rho$,
where $|S|\not=0$. In this case , if any subtree out of the 1 edge
of any variable in $S$ has depth greater than $d-1$ when restricted
to $\rho$, it would imply that $T|_{\rho}$ has depth at least $d$.
In fact, since $|S|\not=0$, even if the subtree out of the 1 edge
of $x_{k}$ has depth greater than $d-1$, it would imply that $T|_{\rho}$
has depth greater than $d$. Thus, at least one of the $|S|+1$ subtrees,
all of which are uncorrelated, should have depth greater than $d-1$. 

2. Let the set $S\subseteq X_{1,t}$ be the variables assigned $*$
by $\rho$, where $|S|\not=0$, and no variable is assigned 1. In
this case, the are $|S|$ subtrees, at least one of which must have
depth greater than $d-1$ for $T|_{\rho}$ to have depth greater than
$d$. 

3. If no variables are assigned $*$ by $\rho$, then if $k$ is the
least index such that $\rho(x_{k})=1$, the subtree out of the 1 edge
of $x_{k}$ must have depth greater than $d$. 

Writing out the sums for each of the three partitions as mentioned
above and summing over least index $k\in\{0,\ldots,t-1\}$ such that
$x_{k+1}=1$ gives the recurrence. \end{proof}
\begin{lem}
\label{lem:induclem}For $r\in\Omega(d2^{t})$, $\gamma_{d}(r)\geq(c_{0}p2^{t})^{d}$. \end{lem}
\begin{proof}
Let $\Delta=c_{0}p2^{t}$ and $\mu_{j}=1-\gamma_{d-1}(r-j)$. The
expression for $\gamma_{d}(r)$ from Lemma \ref{lem:gamm_rec_induc}
is reproduced here. 

\[
\gamma_{d}(r)\geq\sum_{k=1}^{t-1}q\sum_{i=1}^{k}{k \choose i}q^{k-i}p^{i}(1-\mu_{1}^{i+1})+\sum_{i=0}^{t}{t \choose i}q^{t-i}p^{i}(1-\mu_{1}^{i})+\sum_{k=0}^{t-1}q^{k+1}\gamma_{d}(r-1).
\]

\noindent Let 
\begin{eqnarray*}
A(r-1) & = & \sum_{k=1}^{t-1}q\sum_{i=1}^{k}{k \choose i}q^{k-i}p^{i}(1-\mu_{1}^{i+1}),\\
B(r-1) & = & \sum_{i=0}^{t}{t \choose i}q^{t-i}p^{i}(1-\mu_{1}^{i}),\\
\theta & = & \sum_{k=0}^{t-1}q^{k+1}
\end{eqnarray*}
\textbf{Case $t\geq2$:} For this case, we will consider the terms
$A(r-1)$ and $\theta\gamma_{d}(r-1)$. Expanding the expression for
$m$ terms, we get, 
\begin{eqnarray}
\gamma_{d}(r) & \geq & A(r-1)+\theta\gamma_{d}(r-1)\nonumber \\
 & \geq & \sum_{i=1}^{m}\theta^{i-1}A(r-i)+\theta^{m}\gamma_{d}(r-m)\nonumber \\
 & \geq & \sum_{i=1}^{m}\theta^{i-1}A(r-i)\label{eq:mid_eq_1_lowinduc}
\end{eqnarray}
where we used the fact that $\gamma_{d}(r-m)\geq0$. For $1\leq i\leq m-1$,
let $\gamma_{d-1}(r-i)\geq\Delta{}^{d-1}$. Then we have, 
\begin{eqnarray}
A(r-j) & = & \sum_{k=1}^{t-1}q\sum_{i=1}^{k}{k \choose i}q^{k-i}p^{i}(1-\mu_{j}^{i+1})\nonumber \\
 & \geq & \sum_{k=1}^{t-1}q\sum_{i=1}^{k}{k \choose i}q^{k-i}p^{i}(1-\mu_{j})\nonumber \\
 & \geq & \Delta^{d}\frac{42q}{p2^{t}}\sum_{k=1}^{t-1}\sum_{i=1}^{k}{k \choose i}q^{k-i}p^{i}\label{eq:mideq_2_lowind}
\end{eqnarray}

\noindent where in the second line we used $0\leq\mu\leq1$, and in
the third line we used the inductive hypothesis, $\gamma_{d-1}(r-j)\geq\Delta^{d-1}$
for $1\leq j\leq m$. Consider the term
\begin{eqnarray*}
\frac{1}{p}\sum_{k=1}^{t}\sum_{i=1}^{k}{k \choose i}q^{k-i}p^{i} & = & \frac{1}{p}\sum_{k=1}^{t-1}(p+q)^{k}-q^{k}\\
 & = & \frac{1}{p}\sum_{k=1}^{t-1}\frac{1}{2^{k}}\left((1+p)^{k}-(1-p)^{k}\right)\\
 & = & \frac{1}{p}\sum_{k=1}^{t-1}\frac{1}{2^{k}}\sum_{i=0}^{k}{k \choose i}\left(p{}^{i}-(-p)^{i}\right)\\
 & = & \frac{1}{p}\sum_{k=1}^{t-1}\frac{1}{2^{k}}\sum_{i=0,\text{ \ensuremath{i} is odd}}^{k}2{k \choose i}p{}^{i}\\
 & \geq & \sum_{k=1}^{t-1}\frac{2k}{2^{k}}\\
 & = & 2\left(1-\frac{t+1}{2^{t}}\right)\\
 & \geq & \frac{1}{2}
\end{eqnarray*}
where in the first inequality, we used the fact that $p\geq0$ and
picked only the largest term, used Lemma \ref{claim:basic_ineqs}
for the summation in the second last line, and used $t\geq2$ in the
last line. Substituting the above expression in \ref{eq:mideq_2_lowind},
we can write
\[
A(r-j)\geq\Delta^{d}\frac{21q}{2^{t}},
\]
and substituting it in equation \ref{eq:mid_eq_1_lowinduc}, we get
,
\begin{equation}
\gamma_{d}(r)\geq\Delta^{d}\frac{21q}{2^{t}}\left(\frac{1-\theta^{m}}{1-\theta}\right).\label{eq:mideq_3}
\end{equation}

\noindent Note that 
\begin{eqnarray*}
(1-p)^{t} & \geq & \left(1-c_{p}2^{-t}\right)^{t}\geq\frac{1}{2},
\end{eqnarray*}

\noindent and using $q\leq\frac{1}{2}$, we get 
\[
\theta=\frac{q}{1-q}(1-q^{t})\leq1-\frac{(1-\frac{1}{C2^{t}})^{t}}{2^{t}}\leq1-\frac{1}{2\cdot2^{t}}.
\]
Thus for $m=2\cdot2^{t}$, 
\[
\theta^{m}\leq\frac{1}{e},
\]

\noindent and further, 

\begin{eqnarray}
\frac{1-\theta^{m}}{1-\theta} & = & (1-\theta^{m})\frac{1-q}{1-2q+q^{t+1}}\nonumber \\
 & \geq & \frac{1}{2}\frac{\frac{1}{2}}{\frac{1}{C2^{t}}+\frac{1}{2^{t+1}}}\nonumber \\
 & \geq & 2^{t}\left(\frac{\frac{1}{4}}{\frac{1}{2}+\frac{1}{C}}\right)\nonumber \\
 & \geq & 2^{t}\frac{1}{4},\label{eq:theta_bound}
\end{eqnarray}
and substituting it in equation \ref{eq:mideq_3}, and noting that
for $t\geq2$, since $q\geq\frac{1}{2}-\frac{1}{4C}\geq\frac{3}{8}$,
we get,
\[
\gamma_{d}(r)\geq\Delta^{d}\frac{21}{2^{t}}\frac{3}{8}2^{t}\frac{1}{4}\geq\Delta^{d}
\]
as required. \\

\noindent \textbf{Case $t=1$:} For the case of $t=1$, we use the
terms $B(r-1)$ and $\theta\gamma_{d}(r-1)$, and noting that for
$1\leq i\leq m-1$, 
\[
B(r-i)\geq p\Delta^{d-1},
\]
we get by using \ref{eq:theta_bound}, 
\begin{eqnarray*}
\gamma_{d}(r) & \geq & B(r-1)+\theta\gamma_{d}(r-1)\\
 & \geq & \sum_{i=1}^{m}\theta^{i-1}B(r-i)\\
 & \geq & \Delta^{d}\left(\frac{42}{p2}\right)p\left(\frac{1-\theta^{m}}{1-\theta}\right)\\
 & \geq & \Delta^{d}
\end{eqnarray*}
as required.
\end{proof}
In the proof of Lemma \ref{lem:induclem}, we use $O(2^{t})$ steps
to take the induction from $d$ to $d-1$ levels, and the final $d=1$
case can also be carried out with $O(2^{t})$ levels. Thus, the maximum
depth uptil which our conclusions hold is 
\[
d2^{t}\leq O(r)\leq c_{d}\frac{\log n}{\log t}
\]
or 
\[
d\leq c_{d}\left(\frac{\log n}{2^{t}\log t}\right).
\]
This concludes the proof of Theorem \ref{thm:lower_bound}. 
\begin{rem}
If we unroll the induction and see how it worked, for depth $d$,
we find some vertex that is unset by $\rho$ and is connected to a
tree that has depth $d-1$ under the action of $\rho$. For depth
$d-1$, we again find an unset variable connected to a tree of depth
$d-2$. This carries on until the last step, where we want to find
a vertex that has depth greater than or equal to 1, i.e., has a path
to both a 0-leaf and a 1-leaf. Thus, as described in the introduction,
the whole proof essentially finds a \emph{``path with a split''.} 
\end{rem}

\section*{Acknowledgements}

I would like to thank Jaikumar Radhakrishnan, Prahladh Harsha and
Swagato Sanyal for many discussions related to Hastad's switching
lemma.

\bibliographystyle{alpha}
\bibliography{ref}

\appendix

\section{\label{sec:app_low_bound_proofs}Proofs from Section \ref{sec:Lower-Bounds}}
\begin{lem}
\label{lem:app_coeff_consts_rec}\ref{lem:coeff_consts_rec}The recursive
expressions for $[1]P_{0}(r)$ and $[1]P_{1}(r)$ are as follows:
\begin{eqnarray*}
[1]P_{0}(1) & = & \frac{1}{2^{t}}\\{}
[1]P_{1}(1) & = & 1-\frac{1}{2^{t}},
\end{eqnarray*}
and
\begin{eqnarray*}
[1]P_{0}(r) & = & \left(1-\frac{1}{2^{t}}\right)[1]P_{1}(r-1)+\frac{1}{2^{t}},\\{}
[1]P_{1}(r) & = & \left(1-\frac{1}{2^{t}}\right)[1]P_{0}(r-1).
\end{eqnarray*}
\end{lem}
\begin{proof}
From claim \ref{claim:poly_props}, note that the operator $[1]$
is both linear and multiplicative, i.e. 
\begin{eqnarray}
[1](Q+R) & = & [1]Q+[1]R,\nonumber \\{}
[1]QR & = & [1]Q[1]R.\label{eq:=00005B1=00005D_is_lin_mul-1}
\end{eqnarray}
To compute the constant coefficients, we start by the base case, $P_{0}(1)$
and $P_{1}(1)$. 

\[
[1]P_{0}(1)=[1]q^{t}=\frac{1}{2^{t}}
\]
and 
\[
[1]P_{1}(1)=[1]1-\left([1]\frac{1+p}{2}\right)^{t}=1-\frac{1}{2^{t}}.
\]
For the general case of any $r$, using only the equations \ref{eq:=00005B1=00005D_is_lin_mul-1}
and applying it to \ref{eq:p0r}, we get, 
\begin{eqnarray*}
[1]P_{0}(r) & = & [1]\left(qU\sum_{k=0}^{t-1}(q+pU)^{k}+(q+pU)^{t}\right)\\
 & = & [1]q[1]U\sum_{k=0}^{t-1}([1]q+[1]pU)^{k}+([1]q+[1]pU)^{t}\\
 & = & \frac{1}{2}[1]U\sum_{k=0}^{t-1}2^{-k}+2^{-t}\\
 & = & \left(1-\frac{1}{2^{t}}\right)[1]P_{1}(r-1)+\frac{1}{2^{t}}.
\end{eqnarray*}
And similarly, 
\[
[1]P_{1}(r)=\left(1-\frac{1}{2^{t}}\right)[1]P_{0}(r-1).
\]
\end{proof}
\begin{lem}
\label{lem:app_coeff_consts}\ref{lem:coeff_const_exact}Let $\alpha=2^{-t}$
and $\beta=1-\alpha$. Then the exact expressions for $[1]P_{0}(r)$
and $[1]P_{1}(r)$ are as follows. If $r=2k+1$, 
\begin{eqnarray}
[1]P_{0}(r) & = & \frac{1-\beta^{r+1}}{1+\beta},\nonumber \\{}
[1]P_{1}(r) & = & \frac{\beta+\beta^{r+1}}{1+\beta},\label{eq:exact_const_p1_odd-1}
\end{eqnarray}
and if $r=2k$, 
\begin{eqnarray*}
[1]P_{0}(r) & = & \frac{1+\beta^{r+1}}{1+\beta},\\{}
[1]P_{1}(r) & = & \frac{\beta-\beta^{r+1}}{1+\beta}.
\end{eqnarray*}
\end{lem}
\begin{proof}
To compute the exact equations, assume $r=2k+1$. Substituting equation
\ref{eq:const_p1} in equation \ref{eq:const_p0} and using equation
\ref{eq:base_p0} for the base case of $r=1$, we get 
\begin{eqnarray*}
[1]P_{0}(r) & = & \left(1-\frac{1}{2^{t}}\right)^{2}[1]P_{0}(r-2)+\frac{1}{2^{t}}\\
 & = & \left(1-\frac{1}{2^{t}}\right)^{2k}\frac{1}{2^{t}}+\frac{1}{2^{t}}\left(\frac{1-\left(1-\frac{1}{2^{t}}\right)^{2k}}{1-\left(1-\frac{1}{2^{t}}\right)^{2}}\right)\\
 & = & \frac{1-\beta^{r+1}}{1+\beta},
\end{eqnarray*}
And from Lemma \ref{lem:sum_constcoeff_1}, for $r=2k+1$, 
\[
[1]P_{1}(r)=\frac{\beta+\beta^{r+1}}{1+\beta}.
\]
Similarly, for $r=2k$, we get by using one step of equation \ref{eq:const_p0}
followed by \ref{eq:exact_const_p1_odd-1}, 
\[
[1]P_{0}(r)=\frac{1+\beta^{r+1}}{1+\beta},
\]
and similarly, 
\[
[1]P_{1}(r)=\frac{\beta-\beta^{r+1}}{1+\beta}.
\]
\end{proof}
\begin{lem}
\label{lem:app_coeff_p}\ref{lem:coeff_p}The recurrence relations
for $[p]P_{0}(r)$ and $[p]P_{1}(r)$ are as follows:
\begin{eqnarray*}
[p]P_{0}(1) & = & -\frac{t}{2^{t}}\\{}
[p]P_{1}(1) & = & -\frac{t}{2^{t}}.
\end{eqnarray*}
and
\begin{eqnarray*}
[p]P_{1}(r) & = & \left(1-\frac{1}{2^{t}}\right)[p]P_{0}(r-1)+\left(\frac{t+2}{2^{t}}-2\right)[1]P_{0}(r-1)\\
 &  & +2\left(1-\frac{t+1}{2^{t}}\right)[1]P_{0}(r-1){}^{2}\\{}
[p]P_{0}(r) & = & \left(1-\frac{1}{2^{t}}\right)[p]P_{1}(r-1)+\left(\frac{t+2}{2^{t}}-2\right)[1]P_{1}(r-1)\\
 &  & +2\left(1-\frac{t+1}{2^{t}}\right)[1]P_{1}(r-1){}^{2}-\frac{t}{2^{t}}+\frac{2t}{2^{t}}[1]P_{1}(r-1).
\end{eqnarray*}
\end{lem}
\begin{proof}
To evaluate the coefficient of $p$ in the polynomials, first note
that from the claim \ref{claim:poly_props}, the operator $[p]$ acts
as follows on sums and products of polynomials $Q$ and $R$:
\begin{eqnarray}
[p](Q+R) & = & [p]Q+[p]R,\nonumber \\{}
[p]QR & = & [p]Q[1]R+[1]Q[p]R.\label{eq:=00005Bp=00005D_mul-1}
\end{eqnarray}
For the base case of $r=1$, 
\[
[p]P_{0}(1)=[p]q^{t}=-\frac{t}{2^{t}}
\]
and

\begin{eqnarray*}
[p]P_{1}(1) & = & [p]\left(1-(1-q)^{t}\right)=-[p](1-q)^{t}=-\frac{t}{2^{t}}.
\end{eqnarray*}
For any general $r$, we first compute $[p]P_{1}(r)$ using equation
\ref{eq:p1r}. 
\begin{eqnarray*}
[p]P_{1}(r) & = & [p]\left(\sum_{k=0}^{t-1}\sum_{i=0}^{k}{k \choose i}q^{k-i+1}p^{i}\left(P_{0}(r-1)\right)^{i+1}\right)\\
 & = & \sum_{k=0}^{t-1}\sum_{i=0}^{k}{k \choose i}[p]q^{k-i+1}p^{i}\left(P_{0}(r-1)\right)^{i+1}\\
 & = & \sum_{k=0}^{t-1}[p]q^{k+1}P_{0}(r-1)+\sum_{k=0}^{t-1}[1]kq^{k}\left(P_{0}(r-1)\right)^{2}\\
 & = & \sum_{k=0}^{t-1}[1]q^{k+1}[p]P_{0}(r-1)+\sum_{k=0}^{t-1}[p]q^{k+1}[1]P_{0}(r-1)+\sum_{k=0}^{t-1}[1]kq^{k}\left(P_{0}(r-1)\right)^{2}\\
 & = & \left(1-\frac{1}{2^{t}}\right)[p]P_{0}(r-1)+\left(\frac{t+2}{2^{t}}-2\right)[1]P_{0}(r-1)\\
 &  & +2\left(1-\frac{t+1}{2^{t}}\right)[1]P_{0}(r-1){}^{2}
\end{eqnarray*}
where the second and third lines used facts (1), (4) and (5) in claim
\ref{claim:poly_props}, and the fourth line used the fact \ref{eq:=00005Bp=00005D_mul-1}.
Similarly, from equation \ref{eq:p0r}, we get, 
\begin{eqnarray*}
[p]P_{0}(r) & = & [p]\left(\sum_{k=0}^{t-1}\sum_{i=0}^{k}{k \choose i}q^{k-i+1}p^{i}\left(P_{1}(r-1)\right)^{i+1}+\sum_{i=0}^{t}{t \choose i}q^{t-i}p^{i}\left(P_{1}(r-1)\right)^{i}\right)\\
 & = & \left(1-\frac{1}{2^{t}}\right)[p]P_{1}(r-1)+\left(\frac{t+2}{2^{t}}-2\right)[1]P_{1}(r-1)+2\left(1-\frac{t+1}{2^{t}}\right)[1]P_{1}(r-1){}^{2}\\
 &  & +\sum_{i=0}^{t}{t \choose i}[p]q^{t-i}p^{i}\left(P_{1}(r-1)\right)^{i}\\
 & = & \left(1-\frac{1}{2^{t}}\right)[p]P_{1}(r-1)+\left(\frac{t+2}{2^{t}}-2\right)[1]P_{1}(r-1)\\
 &  & +2\left(1-\frac{t+1}{2^{t}}\right)[1]P_{1}(r-1){}^{2}-\frac{t}{2^{t}}+\frac{2t}{2^{t}}[1]P_{1}(r-1).
\end{eqnarray*}
\end{proof}
\begin{lem}
\label{lem:app_p_lowbound}\ref{lem:pcoeff_upperbound}For every $r\geq1$,
$-2\cdot2^{t}\leq[p]P_{0}(r)\leq0$ and $-2\cdot2^{t}\leq[p]P_{1}(r)\leq0$.\end{lem}
\begin{proof}
We reproduce equations \ref{eq:pcoeffp0} and \ref{eq:pcoeffp1} here
for convenience:

\begin{eqnarray}
[p]P_{1}(r) & = & \left(1-\frac{1}{2^{t}}\right)[p]P_{0}(r-1)+\left(\frac{t+2}{2^{t}}-2\right)[1]P_{0}(r-1)\nonumber \\
 &  & +2\left(1-\frac{t+1}{2^{t}}\right)[1]P_{0}(r-1){}^{2}\label{eq:pcoeffp1-1}\\{}
[p]P_{0}(r) & = & \left(1-\frac{1}{2^{t}}\right)[p]P_{1}(r-1)+\left(\frac{t+2}{2^{t}}-2\right)[1]P_{1}(r-1)\nonumber \\
 &  & +2\left(1-\frac{t+1}{2^{t}}\right)[1]P_{1}(r-1){}^{2}-\frac{t}{2^{t}}+\frac{2t}{2^{t}}[1]P_{1}(r-1).\label{eq:pcoeffp0-1}
\end{eqnarray}
We first show that show $[p]P_{0}(r)\leq0$ and $[p]P_{1}(r)\leq0$
by induction on $r$. For $r=1$, from equations \ref{eq:p_p01} and
\ref{eq:p_p11}, $[p]P_{0}(1)\leq0$ and $[p]P_{1}(1)\leq0$. Consider
equation \ref{eq:pcoeffp1-1},

\[
[p]P_{1}(r)=\left(1-\frac{1}{2^{t}}\right)[p]P_{0}(r-1)+\left(\frac{t+2}{2^{t}}-2\right)[1]P_{0}(r-1)+2\left(1-\frac{t+1}{2^{t}}\right)[1]P_{0}(r-1){}^{2}.
\]
For the first term, since $1-2^{-t}\geq0$ and $[p]P_{0}(r-1)\leq0$,
the product is always non-positive. For the third term, $1-2^{-t}(t+1)\geq0$,
and since $0\leq[1]P_{0}(r-1)\leq1$ from Lemma \ref{lem:sum_constcoeff_1},
$[1]P_{0}(r-1)^{2}\leq[1]P_{0}(r-1)$. Thus, we can write,

\begin{eqnarray*}
\left(\frac{t+2}{2^{t}}-2\right)[1]P_{0}(r-1)+2\left(1-\frac{t+1}{2^{t}}\right)[1]P_{0}(r-1){}^{2} & \leq & \left(\frac{t+2}{2^{t}}-2+2-\frac{2(t+1)}{2^{t}}\right)[1]P_{0}(r-1)\\
 & \leq & 0.
\end{eqnarray*}
For $[p]P_{0}(r)$, consider equation \ref{eq:pcoeffp0-1},
\begin{eqnarray*}
[p]P_{0}(r) & = & \left(1-\frac{1}{2^{t}}\right)[p]P_{1}(r-1)+\left(\frac{t+2}{2^{t}}-2\right)[1]P_{1}(r-1)\\
 &  & +2\left(1-\frac{t+1}{2^{t}}\right)[1]P_{1}(r-1){}^{2}-\frac{t}{2^{t}}+\frac{2t}{2^{t}}[1]P_{1}(r-1).
\end{eqnarray*}
The first term is non-positive similar to the argument above, and
summing the remaining terms after using Lemma \ref{lem:sum_constcoeff_1}
for $[1]P_{1}(r-1){}^{2}\leq[1]P_{1}(r-1)$, we get, 
\begin{eqnarray*}
\left(\frac{t+2}{2^{t}}-2\right)[1]P_{1}(r-1)+2\left(1-\frac{t+1}{2^{t}}\right)[1]P_{1}(r-1){}^{2}\\
+\frac{2t}{2^{t}}[1]P_{1}(r-1)-\frac{t}{2^{t}}\leq\frac{t}{2^{t}}[1]P_{1}(r-1)-\frac{t}{2^{t}} & \leq & 0.
\end{eqnarray*}

We now show that show that $[p]P_{0}(r)\geq-2\cdot2^{t}$ and $[p]P_{1}(r)\geq-2\cdot2^{t}$,
again by induction on $r$. For $r=1$, note that
\[
[p]P_{0}(1)=[p]q^{t}=-\frac{t}{2^{t}}\geq-2\cdot2^{t},
\]
and
\[
[p]P_{1}(1)=[p]\left(1-(1-q)^{t}\right)=-[p](1-q)^{t}=-\frac{t}{2^{t}}\geq-2\cdot2^{t}.
\]
Consider equation \ref{eq:pcoeffp1-1}. For $t\geq1$, $1-\alpha(t+1)\geq0$,
$\alpha(t+2)-2\leq0$, and $0\leq P_{0}(r-1)\leq1$, and thus, we
can write using the induction hypothesis,
\begin{eqnarray*}
[p]P_{1}(r) & = & \beta[p]P_{0}(r-1)+\left(\alpha(t+2)-2\right)[1]P_{0}(r-1)+2\left(1-\alpha(t+1)\right)[1]P_{0}(r-1){}^{2}\\
 & \geq & -2\beta2^{t}+\alpha\left(t+2-2\cdot2^{t}\right)\\
 & \geq & -2(\beta+\alpha)2^{t}\\
 & = & -2\cdot2^{t}.
\end{eqnarray*}
In equation \ref{eq:pcoeffp0-1}, noting that $0\leq P_{1}(r-1)\leq1$,
we can similarly write
\begin{eqnarray*}
[p]P_{0}(r) & = & \left(1-\frac{1}{2^{t}}\right)[p]P_{1}(r-1)+\left(\frac{t+2}{2^{t}}-2\right)[1]P_{1}(r-1)+2\left(1-\frac{t+1}{2^{t}}\right)[1]P_{1}(r-1){}^{2}\\
 &  & -\frac{t}{2^{t}}+\frac{2t}{2^{t}}[1]P_{1}(r-1)\\
 & \geq & -2\beta2^{t}+\alpha(t+2-2\cdot2^{t})-\alpha t\\
 & \geq & -2\cdot2^{t}
\end{eqnarray*}
as required. 
\end{proof}

\end{document}